\documentclass[11pt]{article}
\usepackage[utf8]{inputenc}
\usepackage{amssymb,amsfonts,amsthm}
\usepackage{amsmath}
\usepackage{xcolor}
\usepackage{makecell}
\usepackage{geometry}
\usepackage{tikz}
\usepackage{algorithm}
\usepackage{algorithmicx}
\usepackage{algpseudocode}
\usepackage{fullpage}
\usepackage{subcaption} 

\usetikzlibrary{decorations.pathreplacing,angles,quotes}
\usetikzlibrary{arrows.meta}

\usepackage{hyperref}
\usepackage{xcolor}
\hypersetup{
  colorlinks   = true, 
  urlcolor     = blue, 
  linkcolor    = blue, 
  citecolor   = blue 
}

\usepackage{geometry}
\newcommand*\samethanks[1][\value{footnote}]{\footnotemark[#1]}

\newcommand{\E}{\mathbb{E}}
\newcommand{\aff}{\text{aff}}

\newcommand{\NAND}{\text{NAND}}
\newcommand\defeq{\mathrel{\stackrel{\makebox[0pt]{\mbox{\normalfont\tiny def}}}{=}}}

\newcommand{\R}{\mathbb{R}}
\newcommand{\row}{\mathsf{row}}
\renewcommand{\ker}{\mathsf{ker}}
\newcommand{\TV}{d_{\mathrm{TV}}}

\newtheorem{theorem}{Theorem}
\newtheorem{lemma}{Lemma}
\newtheorem{corollary}{Corollary}[]

\theoremstyle{definition}
\newtheorem{definition}{Definition}[section]
\newtheorem{remark}{Remark}[]

\makeatletter
\newcommand{\neutralize}[1]{\expandafter\let\csname c@#1\endcsname\count@}
\makeatother

\newenvironment{lembis}[1]
  {\renewcommand{\thelemma}{\ref{#1}$'$}%
  \neutralize{lemma}\phantomsection
  \begin{lemma}}
  {\end{lemma}}
  

  
\title{Linear Programs with Polynomial Coefficients \\ and Applications to 1D Cellular Automata
}
\author{Guy Bresler \thanks{ Massachusetts Institute of Technology, email: \{guy, chenghao, yp\}@mit.edu }, Chenghao Guo\samethanks, Yury Polyanskiy\samethanks}
\date{}

\begin{document}
\maketitle
\def\nbyp#1{{\color{blue}[\textbf{YP:} #1]}}

\begin{abstract}
        Given a matrix $A$ and vector $b$ with polynomial entries in $d$ real variables $\delta=(\delta_1,\ldots,\delta_d)$ we consider the following notion of feasibility: the pair $(A,b)$ is \emph{locally feasible} if there exists an open neighborhood $U$ of $0$ such that for every $\delta\in U$ there exists $x$ satisfying $A(\delta)x\ge b(\delta)$ entry-wise. 
        For $d=1$ we construct a polynomial time algorithm for deciding local feasibility. For $d \ge 2$ we show local feasibility is NP-hard. This also gives the first polynomial-time algorithm for the asymptotic linear program problem introduced by \cite{jeroslow1973asymptotic}.
        
	As an application (which was the primary motivation for this work) we give a
	computer-assisted proof of ergodicity of the following elementary 1D cellular automaton: given
	the current state $\eta_t \in  \{0,1\}^{\mathbb{Z}}$ the next state $\eta_{t+1}(n)$ at each vertex $n\in \mathbb{Z}$ is obtained by $\eta_{t+1}(n)= \NAND\big(\text{BSC}_\delta(\eta_t(n-1)),  \text{BSC}_\delta(\eta_t(n))\big)$. Here the binary symmetric channel $\text{BSC}_\delta$ takes a bit as input and flips it with probability $\delta$ (and leaves it unchanged with probability $1-\delta$). 
	It is shown that there exists
	$\delta_0>0$ such that for all $0<\delta<\delta_0$ the distribution of $\eta_t$ converges to
	a unique stationary measure irrespective of the initial condition $\eta_0$. 
        
    We also consider the problem of broadcasting information on the 2D-grid of noisy binary-symmetric channels $\text{BSC}_\delta$, where each node may apply an arbitrary processing function to its input bits. We prove that there exists
	$\delta_0'>0$ such that for all noise levels $0<\delta<\delta_0'$ it is impossible to broadcast information
    for any processing function, as conjectured in~\cite{makur2020broadcasting}.

\end{abstract}


\tableofcontents

\newpage

\section{Introduction}
\pagenumbering{arabic}

Linear Programming (LP) is one of the central paradigms of optimization, with broad significance in both theory and applications. In this paper we introduce and study a version of linear programming, \emph{polynomial linear programming (PLP)}, where the coefficients are given by polynomials. 


Our initial motivation in studying PLPs was in trying to generalize the exciting method of proving ergodicity of cellular automata introduced in~\cite{holroyd2019percolation}. That method relies on finding a certain potential (or a Lyapunov function) that decreases on average. The problem of finding such potential was shown in~\cite{makur2020broadcasting} to be a PLP. 

However, PLPs also arise naturally. As an example, consider the max-flow problem, which (as is known classically) can be formulated as an LP. A typical question is whether a given graph with fixed  edge capacities can support at least a flow $r$ between a chosen source-sink pair. If the edge capacities and desired flow depend on some external factors $\delta$, e.g. the amount of rainfall in the case of traffic modeling, then the problem may be approximated by a PLP. The question would then be whether the flow $r(\delta)$ is possible even if the amount of rainfall varies slightly. This robustness question is what we call local-feasibility of a PLP and is the subject of this work. 
This notion of robustness is distinct from the one considered in the field of robust optimization, since (as elaborated upon in Section~\ref{sec:related-literature}) robust optimization aims for a single solution $r$ (not depending on $\delta$) and the uncertainty sets are not algebraic manifolds as we consider.

We now formally define PLPs and the associated computational problems. 


\begin{definition}[$d$-PLP]
    A \emph{$d$-dimensional polynomial linear program ($d$-PLP)} is specified by an $m\times n$ matrix $A(\delta)$ and an $m\times 1$ vector $b(\delta)$ of degree-$D$ polynomials in $\delta\in \mathbb{R}^d$. The coefficients of all polynomial functions are rational numbers represented by integer numerators and denominators. The \emph{size} of a PLP is defined to be the number of bits needed to describe $A$ and $b$ (see Definition~\ref{def:size-of-input} for full detail).
\end{definition}


We are interested in two types of PLP feasibility: local and everywhere feasibility.  

\begin{definition}[Local feasibility]
    The PLP $\left(A(\delta),b(\delta)\right)$ is said to be \emph{locally feasible} \emph{(locally infeasible)} if there exists $\epsilon>0$ such that for every $\delta \in B_\epsilon (0)$ there exists (does not exist) an $x\in\mathbb{R}^n$ such that $A(\delta)x\ge b(\delta)$.
\end{definition}
\begin{definition}[Everywhere feasibility]
    The PLP $(A(\delta),b(\delta))$ is said to be \emph{everywhere feasible} \emph{(everywhere infeasible)} if for every $\delta \in \mathbb{R}^d$ there exists (does not exist) an $x\in\mathbb{R}^n$ such that $A(\delta)x\ge b(\delta)$.
\end{definition}

Intuitively, the local and everywhere feasibility of PLP describes whether a linear program is ``stable'' subject to changes in the parameters. Imagine a system having a property we care about that is determined by the feasibility of an LP $Ax\ge b$, but the system is not known exactly and we are only sure that coefficients of the LP lie on a polynomial manifold $(A(\delta),b(\delta))$, parameterized by $\delta$. The everywhere feasibility then tells us whether the property is satisfied up to arbitrary perturbation on the manifold, while the local feasibility tells us whether the property is robust locally around a fixed point.

\begin{remark}
For local feasibility, by a simple translation, taking $0$ as the center of the local neighborhood ball is computationally equivalent to any other constant point in $\mathbb{R}^d$.
\end{remark}

\begin{remark}\label{rmk:interval}
When $d=1$, everywhere feasibility on $\mathbb{R}$ is equivalent to feasibility on any interval $(c,d)$. To transform a problem on $\mathbb{R}$ to the interval $(c,d)$, we can substitute $c+\frac{\delta^2(d-c)}{\delta^2+1}$ for $\delta$ and multiply by the denominator on both sides. Similarly, to transform a problem on the interval $(c,d)$ to $\mathbb{R}$, we can replace $\delta$ with $1/(c-\delta)+1/(d-\delta)$ and multiply by the denominator on both sides. 
\end{remark}


\subsection{Results}

We close the tractability of both types of PLP feasibility in all dimensions, and apply the results to various problems.

\begin{theorem}[Local feasibility of 1-PLP is tractable]\label{thm:1-dim-local}
The local feasibility of 1-PLP can be solved in polynomial time. 
Furthermore, if the PLP is locally feasible, a solution $x(\delta)$ can be given as a rational number at 0 and two rational functions of $\delta$ on the positive and negative neighborhoods of 0, respectively. The interval $(\delta_1,\delta_2)$ on which $x(\delta)$ is a feasible solution is given by $\delta_1$ and $\delta_2$ each being the smallest positive root of an explicit polynomial. The running time is $O(n^3mD^3\cdot  \mathsf{LP}(O(n^2D),O(nmD),L))$, where $\mathsf{LP}(a,b,c)$ stands for the running time of ordinary linear programming with $a$ variables, $b$ constraints and $c$-bit numbers.
\end{theorem}

\begin{theorem}[Hardness]\label{thm:hardness}
The everywhere feasibility of $d$-PLP is NP-hard for any $d$ and is co-NP-complete when $d=1$. The local feasibility of $d$-PLP is NP-hard for any $d\ge 2$. 
\end{theorem}

As a consequence of Theorem~\ref{thm:1-dim-local}, the max-flow problem discussed in the beginning, where capacity and demand are polynomials in $\delta$ which is 1-dimensional and varies locally around $0$, can be solved in polynomial time by viewing it as a special case of PLP. 
However, it is not clear whether the max-flow problem is tractable when $\delta$ has higher dimension: The hardness construction of Theorem~\ref{thm:hardness} is an LP that cannot be easily reduced to a flow problem. We leave this question open.

Aside from direct application to problems such as max-flow, the algorithm in Theorem~\ref{thm:1-dim-local} will 
be used to give computer-assisted proofs solving open problems in probabilistic cellular automata
(PCA) and broadcasting on the 2-dimensional grid under the most interesting regime of vanishing (but non-zero) noise. See Section~\ref{sec:applications} for a summary of these results.
\begin{theorem}[Informal]\label{thm:informal-pca} There exists $\delta_0>0$ such that an 
    elementary PCA with NAND function and either vertex binary-symmetric-channel (BSC) noise
    $\delta$ or edge BSC noise $\delta$ is ergodic for all $\delta \in (0,\delta_0)$. 
\end{theorem}
\begin{theorem}[Informal]\label{thm:informal-broadcasting} There exists $\delta_0>0$ such that 
    when information is broadcast on the space-homogeneous 2D grid of relays interconnected by
    BSC$(\delta)$, it is impossible to
    recover any information about the origin from the boundary values for all $\delta \in
    (0,\delta_0)$.
\end{theorem}

The $\text{BSC}_p$ channel has binary input and output. For a given noise level $p\in [0,1/2]$, it outputs the same bit as input with probability $1-2p$ and outputs a random bit with probability $2p$.

\begin{remark}
In Theorems~\ref{thm:informal-pca} and~\ref{thm:informal-broadcasting} we are interested in only the feasibility on the positive neighborhood of 0 rather than $B_\epsilon(0)$, which is different than the statement of Theorem~\ref{thm:1-dim-local}. However, as discussed in the beginning of Section~\ref{sec:algo-dimension-1}, our algorithm considers the positive side and negative side separately and combines the results. So the algorithm can also solve feasibility on one side.
\end{remark}

\begin{remark}\label{rmk:small_noise}
We note that impossibility of broadcasting via noisy NANDs (or any other symmetric function) was
conjectured in~\cite{makur2020broadcasting} for all $\delta\in(0,1/2]$. Our result above resolves
this conjecture only partially. Indeed, intuitively showing impossibility of broadcasting
with a small noise (arbitrarily close to $0$) as we do is the most
challenging setting. However, there are no known monotonicity (in the noise magnitude) results for
these problems, so it remains possible in principle that an intermediate amount noise could somehow help information propagate. 

In the high noise regime, Evans-Schulman~\cite{evans1999signal} show that for $\delta >
{2-\frac{\sqrt{2}}{4}}\approx 0.146$ broadcasting is impossible with any arrangement of two-input
gates. To improve this estimate for the graph corresponding to a 2D grid we can use a) the connection between information and percolation from~\cite[Theorem
5]{polyanskiy2017strong} and b) a rigorous lower bound $p_{lb} = 0.6231$ for the oriented bond-percolation
threshold from~\cite{gray1980lower}. These arguments show that broadcasting on a 2D grid of
arbitrary gates and $\text{BSC}_\delta$ wires is impossible for $\delta > {1-\frac{\sqrt{p_{lb}}}{2}}
\approx 0.105$.
It remains open to close the gap between our results for small $\delta<\delta_0$ and these ones for large $\delta>\delta_1$. Our method cannot do so, because it is based on feasibility of general PLP: 
Showing 
impossibility of broadcasting (or ergodicity of the PCA) in the 
interval $(\delta_0,\delta_1)$ for some given fixed $\delta_1$ requires (by
Remark~\ref{rmk:interval})
solving everywhere feasibility of a PLP, which by Theorem~\ref{thm:hardness} is NP-hard in general. 
Nevertheless, the following procedure gives a potential approach of checking feasibility on $(\delta_0,\delta_1)$ via algorithm stated in Theorem~\ref{thm:1-dim-local}. 
Start by running the algorithm at $\delta_0$, repeatedly run the algorithm at the right most endpoint that guarantees feasibility returned by the previous run. 
If the right endpoint finally exceeds $\delta_1$, then the PLP is feasible on the whole interval $(\delta_0,\delta_1)$.

\subsection{From Feasibility to Optimization}



In this section we will show that locally optimizing PLP can be reduced to checking local feasibility of a PLP. Indeed, this is simply a consequence of the fact that finding a solution of a linear program is (informally) equivalent to finding a feasible solution to the set of KKT conditions. To formalize this intuition, we first define the optimization counter part of local feasibility.

\begin{definition}[Local 1-PLP]\label{def:local-1-plp}
Let $c(\delta)$, $A(\delta)$ and $b(\delta)$ be polynomials of $\delta\in \mathbb{R}$. Let $P(\delta)$ be the following linear programming problem parameterized by $\delta$.
    \[
\text{maximize } c(\delta)^\top x(\delta), \text{ subject to } A(\delta)x\le b(\delta) \text{ and } x(\delta)\ge 0\,.
\]
A \emph{local 1-PLP} seeks to either (1) certify the existence of $\delta_0$ such that for any $0<\delta<\delta_0$, $P(\delta)$ is infeasible, (2) certify the existence of $\delta_0$ such that for any $0<\delta<\delta_0$, $P(\delta)$ is unbounded, or (3) find
a rational function $x(\delta)$ so that there exists $\delta_0>0$ such that for all $0<\delta< \delta_0$ the value $x(\delta)$ is the maximizer of the program.
\end{definition}
It is not immediate that the three cases cover all possibilities. This will be shown in the proof of the Theorem~\ref{thm:local-1-plp}, that demonstrates also how algorithms in Theorem~\ref{thm:1-dim-local} and Theorem~\ref{thm:1-dim-local-infea} can be turned into Local 1-PLP solvers.

This problem was first studied by Jeroslow \cite{jeroslow1973asymptotic,jeroslow1973linear} under the name \emph{asymptotic linear program}. 
In particular they studied the following optimization problem 
\[
\text{maximize } c(M)^\top x(M), \text{ subject to } A(M)x(M)\le b(M) \text{ and } x\ge0\,.
\]
The goal is to find a rational function $x(M)$ of a scalar $M\ge 0$ so that there exists $M_0>0$ such that for all $M \ge M_0$ the value $x(M)$ is the maximizer of the program (such $x(M)$ is proven to exists in \cite{jeroslow1973asymptotic}).
By taking $\delta=1/M$, this is equivalent to local 1-PLP defined in Definition~\ref{def:local-1-plp}, therefore, the results in this section also applies to asymptotic linear program. The key observation in \cite{jeroslow1973asymptotic} is that all rational functions form an ordered field $R(M)$, where the order $>0$ is defined by the asymptotic value as $M$ goes to infinity. 
It was shown that linear programming over any ordered field, including $R(M)$, can be solved using the simplex method, where all basic operations of the simplex are defined over the field.  However, the worst case running time of such algorithm is exponential in the size of input. Indeed, even when $A$ and $b$ do not vary with $M$, there are examples where the simplex method runs for exponential time. 
In contrast, Theorem~\ref{thm:local-1-plp} will give a polynomial-time algorithm for the equivalent problem. 

\begin{theorem}[Local 1-PLP is Tractable]\label{thm:local-1-plp}
    There exists a polynomial time algorithm that solves a local 1-PLP.
\end{theorem}
\begin{proof}
Let $P(\delta)$ be defined as in Definition~\ref{def:local-1-plp}. 
Let $D(\delta)$ be the dual of $P(\delta)$:
\[
\text{minimize } b(\delta)^\top y(\delta), \text{ subject to } A(\delta)^\top y(\delta)\ge c(\delta) \text{ and } y(\delta)\ge 0\,.
\]
From the weak and strong duality theorem, we have the following equivalence relations for any fixed $\delta$:
\[\text{$P(\delta)$ being infeasible} \Leftrightarrow\  A(\delta)x(\delta)\le b(\delta), x(\delta)\ge 0 \text{ being infeasible}\]
\[\text{$P(\delta)$ being unbounded} \Leftrightarrow\ \begin{cases}
A(\delta)x(\delta)\le b(\delta), x(\delta)\ge 0 \text{ being feasible }\\
A(\delta)^\top y(\delta)\ge c(\delta) , y(\delta)\ge 0 \text{ being infeasible}
\end{cases}
\]
\begin{align*}
&\text{$P(\delta)$ being bounded and optimized by $x(\delta)$}\\
&\Leftrightarrow\ \exists x(\delta),y(\delta) \quad \begin{cases}
A(\delta)x(\delta)\le b(\delta)\\
x(\delta)\ge 0\\
A(\delta)^\top y(\delta)\ge c(\delta) \\
y(\delta)\ge 0\\
c(\delta)^\top x(\delta)\ge b(\delta)^\top y(\delta)
\end{cases}
 \text{ being feasible with solution $x(\delta)$}
\end{align*}

By taking $\delta$ close to 0, we have the following similar relations by definition:
\[\text{$P(\delta)$ being locally infeasible at $0^+$} \Leftrightarrow\  A(\delta)x(\delta)\le b(\delta), x(\delta)\ge 0 \text{ being locally infeasible at $0^+$}\]
\[\text{$P(\delta)$ being locally unbounded at $0^+$} \Leftrightarrow\ \begin{cases}
A(\delta)x(\delta)\le b(\delta), x(\delta)\ge 0 \text{ being locally feasible at $0^+$}\\
A(\delta)^\top y(\delta)\ge c(\delta) , y(\delta)\ge 0 \text{ being locally infeasible at $0^+$}
\end{cases}
\]
\begin{align*}
&\text{$P(\delta)$ being locally bounded and optimized by $x(\delta)$ at $0^+$}\\
&\Leftrightarrow\ \begin{cases}
A(\delta)x(\delta)\le b(\delta)\\
x(\delta)\ge 0\\
A(\delta)^\top y(\delta)\ge c(\delta) \\
y(\delta)\ge 0\\
c(\delta)^\top x(\delta)\ge b(\delta)^\top y(\delta)
\end{cases}
 \text{ being locally feasible with solution $x(\delta)$ at $0^+$}
\end{align*}
Here ``at $0^+$'' means there exists $\delta_0>0$ such that for any $0<\delta<\delta_0$ the regarding statement holds.
Note that the left hand sides are exactly the definition of local 1-PLP and all the problems on the right hand sides are problems stated in Theorem~\ref{thm:1-dim-local} and Theorem~\ref{thm:1-dim-local-infea}, local 1-PLP can be solved in polynomial time.
\end{proof}

%
%
\end{remark}


\subsection{Toy Example of 1-PLP} 
To give intuition for how the feasibility of PLP can vary with $\delta$, we start with a toy example of a 1-PLP. This 1-PLP is feasible when $\delta\notin (0,\epsilon)$ and infeasible when $\delta\in (0,\epsilon)$, where $\epsilon$ is an arbitrary small constant.
This example demonstrates that local feasibility cannot be directly solved by simply taking a representative $\delta$ close enough to 0, as $\epsilon$ can be too small for any particular choice of $\delta$. (While here the dependence on $\epsilon$ is simple, in general it is not clear how to determine what constitutes a small enough value of $\delta$ based on the constraints.)

We take  $n=2$, writing $x$ and $y$ as the variables. There are 5 constraints $c_1$ to $c_5$, and only $c_4$ and $c_5$ are varying with $\delta$. The parameter $\epsilon$ is left free for now, to be thought of as a small constant. The 1-PLP is given by the constraints
\begin{align}\label{eq:example-PLP}
\begin{split}
    c_1(x,y)=y-x&\ge 0\\
    c_2(x,y)=y+x&\ge 0\\
    c_3(x,y)=y-1+\epsilon&\le 0\\
    c_4(x,y)=y-\delta^2+\epsilon\delta&\le 0\\
    c_5(x,y)=y+1-\delta&\ge 0\,.
\end{split}
\end{align}
Now consider varying $\delta$ from $-\infty$ to $+\infty$. We depict the evolution of feasibility in Figure~\ref{fig:evolution-example}. 

\begin{figure}[t]
     \centering
     \begin{subfigure}[b]{0.24\textwidth}
         \centering
         \includegraphics[width=\textwidth]{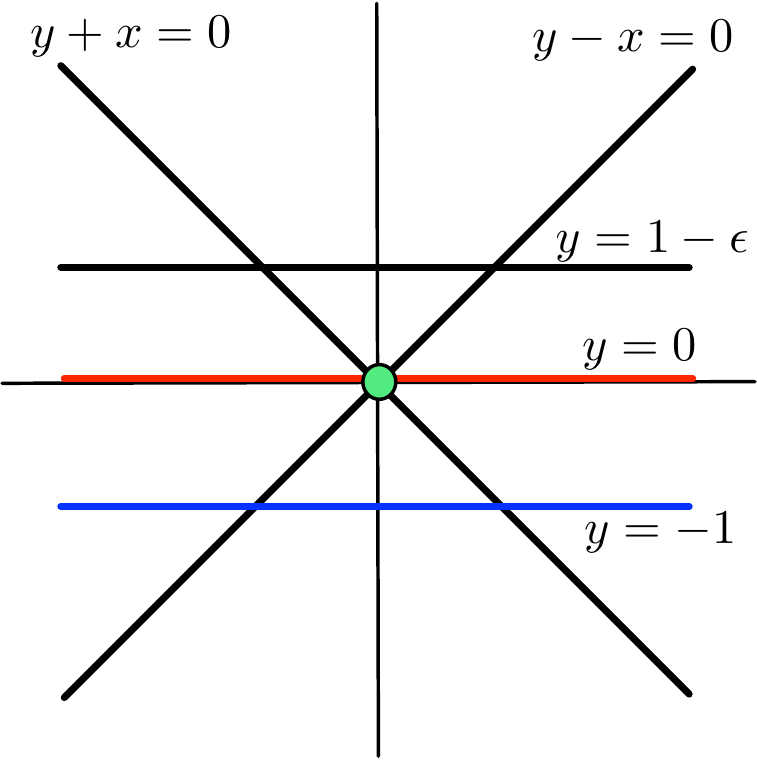}
         \caption{$\delta=0$}
     \end{subfigure}
     \hfill
     \begin{subfigure}[b]{0.24\textwidth}
         \centering
         \includegraphics[width=\textwidth]{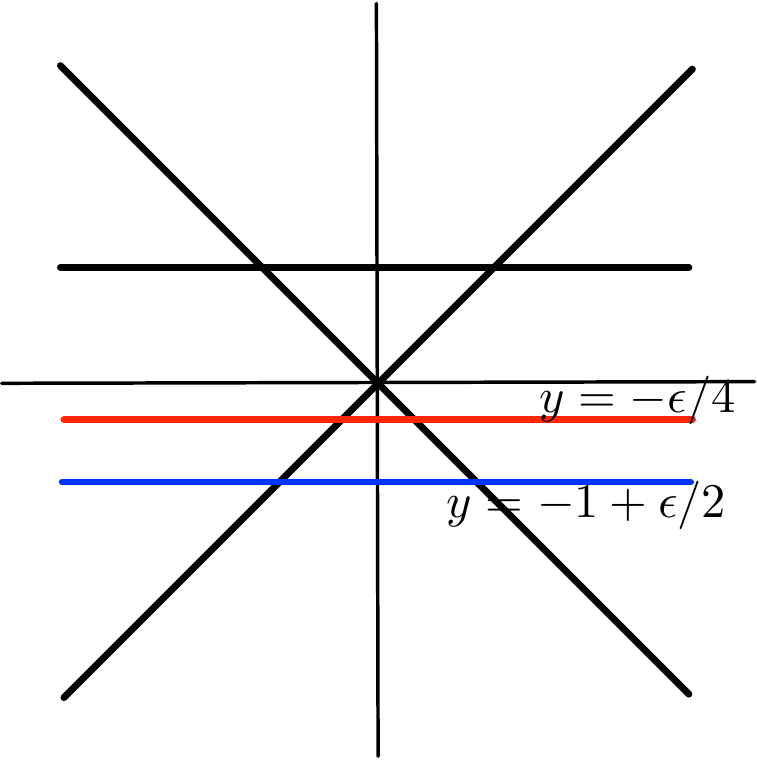}
         \caption{$\delta=\epsilon/2$, no solution}
     \end{subfigure}
     \hfill
     \begin{subfigure}[b]{0.24\textwidth}
         \centering
         \includegraphics[width=\textwidth]{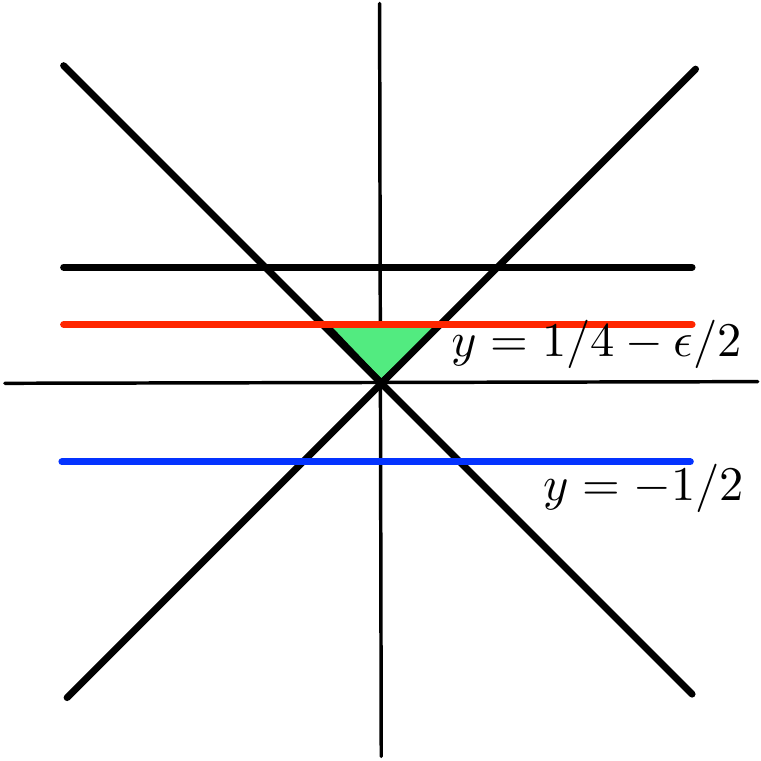}
         \caption{$\delta=1/2$}
     \end{subfigure}
     \hfill
     \begin{subfigure}[b]{0.24\textwidth}
         \centering
         \includegraphics[width=\textwidth]{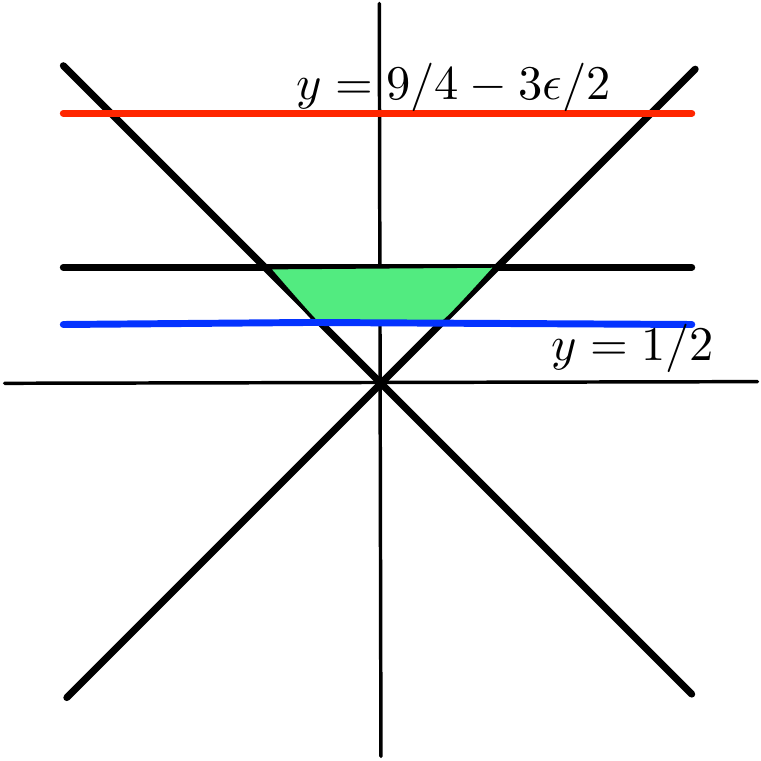}
         \caption{$\delta=3/2$}
     \end{subfigure}
        \caption{Feasible region of the PLP \eqref{eq:example-PLP} as $\delta$ changes. When $\delta\le 0$, it is feasible, and when $\delta=0$, there is a unique solution $x=y=0$. When $\delta\in (0,\epsilon)$, there is no solution because the combination of $c_1\ge 0, c_2\ge 0$, and $c_4\le 0$. When $\delta\in (\epsilon,1]$, the solution set is the region between $c_1=0,c_2=0$, and $c_4=0$, and taking any two of these constraints as equality constraints yields a feasible solution. When $\delta>1$, the solution set is the region between $c_1=0,c_2=0,c_3=0$ and $c_5=0$.}
        \label{fig:evolution-example}
\end{figure}

\subsection{A Polynomial-time Algorithm by Choosing a Small Enough $\delta$}


As described in Theorem~\ref{thm:1-dim-local}, if a PLP is locally feasible, there exists a rational-function solution and $\epsilon>0$ that  satisfies \eqref{equ:lp} for any $\delta\in (0,\epsilon)$. A natural idea is to choose a $\delta_0<\epsilon$ and verify the feasibility at a single point $\delta_0$ (which is a plain LP). If the number of bits in $\delta_0$ is polynomial, this method yields a polynomial-time algorithm for testing local feasibility. 

\begin{lemma}\label{lem:testing-small-delta}
    There exists $\epsilon>0$ such that
    PLP described by \eqref{equ:lp}
    is feasible for all $\delta$ in $(0,\epsilon)$, if and only if \eqref{equ:lp} is feasible when $\delta$ equals $\delta_0=2^{-8n(L+\log Dn)-1}$.
\end{lemma}

Proof for Lemma~\ref{lem:testing-small-delta} is put in the end of Section~\ref{sec:rational-function-solution}, as it uses ideas from the proof of Theorem~\ref{thm:1-dim-local}. 
Although this gives a polynomial time algorithm for testing local feasibility in theory, 
the Lemma demonstrates that the precision of the arithmetic required for this method is infeasiable for some worst case. Furthermore, in practice, we observed that this method does not result in a practical algorithm for the application tasks we care about.
For the applications in Theorem~\ref{thm:informal-pca} and \ref{thm:informal-broadcasting}, $\delta_0$ needs to be less than $2^{-20000}$. 
Also, this approach only gives a 1-bit answer of whether the PLP is locally feasible without a way to verify the answer, which is especially important in light of possible numerical errors.
In contrast, the algorithm we are using in Section~\ref{s:OutputFeas}
outputs a solution given as a function of $\delta$ with rational-numbered coefficients, allowing us to verify the validity of the output by examining derivatives of the constraints at 0.

\subsection{Sketch of Main Ideas}

Here we present some of the ideas underlying  the algorithm in Theorem~\ref{thm:1-dim-local}.

\paragraph{Reduction to Polynomial Solutions.}
Suppose a given PLP is locally feasible. It is not obvious that the solution $x$ can be expressed explicitly: $x(\delta)$ can potentially be any function defined on $B_\epsilon(0)$, and the space of such functions has infinite dimension. In Figure~\ref{fig:evolution-example}, we observe that $x(\delta)$ can be taken to be a piece-wise rational function of $\delta$. 
This turns out to be true in general, and by carefully using the structure of the problem we will prove that there must exist a rational function $x(\delta)$ locally solving the PLP if it is locally feasible. To this end, we demonstrate that each corner (or point meeting a subset of constraints with equality) of the feasible polyhedron is a rational function in $\delta$ on a small enough open interval, and then argue that feasibility of the point cannot change on the interval using a continuity argument. 

We then show that this same statement holds for $x(\delta)$ having polynomial entries, with only a polynomial increase in the magnitude of the coefficients and the degrees.
As a result, we can restrict our search space to $x(\delta) = (x_1(\delta),\dots, x_n(\delta))$ being a vector of polynomials, written as 
$$
x_i(\delta) = \sum_{j} h_{ij} \delta^j \,,
$$
and view the PLP as a set of constraints on the coefficients $(h_{ij})_{ij}$ of the polynomials. Crucially, the dimension of the search space has been changed from infinite to finite.

\paragraph{Solving Feasibility Problem over Coefficients of Polynomials.}

Given the reduction to polynomial solutions described just above, we now investigate the form of the constraints on the coefficients $(h_{ij})_{ij}$ of the polynomial entries of $x(\delta)$.
It turns out to be much easier to understand the constraints in a variant of local feasibility of the PLP where we restrict attention to the positive axis $\delta>0$. The latter problem is equivalent to each of the $m$ constraint polynomials $f_i(\delta)\defeq a_i(\delta)x(\delta)-b_i(\delta)$ being either $0$ or having its first non-zero derivative (in $\delta$) be positive at zero. Let us call this condition the ``derivative condition". 

We now express the derivative condition geometrically.
Observe that $f_i(\delta)$ is a linear function of the $(h_{ij})_{j}$, hence so too is any order of derivative of $f_i(\delta)$.
Each derivative being zero or positive is therefore a linear constraint on $(h_{ij})_{j}$ plus a halfspace constraint.
%
The derivative condition on a particular $f_i(\delta)$ thus corresponds geometrically to a set $S_i$ given by a union of half affine subspaces (halfspace intersected with an affine subspace), so the derivative condition is described by the intersection of the $S_i$.  

Local feasibility for $\delta$ boils down to whether the intersection of the $S_i$ sets is nonempty, in which case there exists a choice of coefficients $(h_{ij})_{ij}$ satisfying the derivative constraints. 
This problem will in turn be reduced to the \emph{subspace elimination problem}, which asks whether a union of half affine spaces consists of the entire ambient space $\mathbb{R}^N$. 

Subspace elimination is introduced in Section~\ref{sec:algo-dimension-1}, and we show how to exploit the underlying geometry of the problem in order to solve it with ordinary linear programming as a subroutine. At a high level, the algorithm iteratively processes and removes subspaces that have the potential to change the dimension of the current feasible set, which is updated after each step. Crucially, we can efficiently determine whether a subspace is relevant based on the dimension of its intersection with the current feasible set.



\paragraph{The Obstruction in Dimension $d\geq 2$.}
A similar approach does not work for local feasibility when $\delta$ has higher dimension than $1$ and indeed we show that the problem is NP-hard. We now give an example illustrating the root of the issue. Suppose $\delta=(\delta_1,\delta_2)$ is 2-dimensional. It is possible that a PLP is locally feasible but there is no single rational function solution $x(\delta)$ which satisfies the program for all $\delta\in B_{\epsilon}(0)$. Consider the following simple example consisting of a single constraint:
\begin{equation*}
    \delta_1\delta_2(\delta_1+\delta_2)x=\delta_1^2\delta_2^2.
\end{equation*}
For any choice of $\delta_1$ and $\delta_2$, the linear equation has a solution. But $x=1/\delta_1+1/\delta_2$ is only feasible for $\delta_1\not=0$ and $\delta_2\not=0$. This example can be generalized so that even locally around the origin one must consider a partition into many regions with each rational function only working on a subset.

\paragraph{Computational Hardness in Dimension $d\geq 2$.}
To prove the hardness result, Theorem~\ref{thm:hardness}, we first show that the everywhere feasibility of 1-PLP can be reduced to the local feasibility of a 2-PLP by a mapping from $\mathbb{R}$ to rays on $\mathbb{R}^2$. So it suffices to prove the hardness for everywhere feasibility of 1-PLP,
$$\forall \delta\in \mathbb{R}, \exists x\in\mathbb{R}^n, A(\delta)x\ge b(\delta).$$
To make the format more aligned with a typical optimization problem that asks whether a solution exists, we can take the negation and LP dual, which gives us
\begin{equation*}
    \exists \delta,y \text{ satisfying }y\ge 0 \text{ and } A^\top (\delta)y=0\text{ with }b^\top (\delta)y> 0\,.
\end{equation*}
We will show through transformation of variables that this problem can be used to approximate an arbitrary \emph{integer linear programming} instance
$$\exists x\in\{0,1\}^n,\ Ax\ge b$$
with arbitrary precision. 
The precision of the approximation is enough to guarantee equivalence between any instance of the maximum independent set problem and everywhere feasibility of a polynomial-size instance of 1-PLP.
Hardness of everywhere feasibility of 1-PLP then follows from NP-hardness of maximum independent set.

\subsection{Outline}
We review related literature in Section~\ref{sec:related-literature}, including cylindrical algebraic decomposition and  robust optimization. In Section~\ref{sec:additional-results}, we give several additional results that are closely related to local feasibility and everywhere feasibility. Results for applications to PCA and broadcasting on 2D grid are in Section~\ref{sec:applications}. Section~\ref{sec:preliminary} contains preliminaries and notation.

The proof and algorithm for Theorem~\ref{thm:1-dim-local} is in Section~\ref{sec:algo-dimension-1} and the proof of Theorem~\ref{thm:hardness} is in Section~\ref{sec:hardness}. Detailed proofs and discussion for applications to PCA and broadcasting on 2D grid are given in Section~\ref{sec:PCA} and Section~\ref{sec:broadcasting-on-grid}.

\subsection{Related Literature}\label{sec:related-literature}
\paragraph{Applications of Asymptotic Linear Program}

Asymptotic linear program has many applications. In \cite{hordijk1985sensitivity,altman1999asymptotic}, it is used to solve policy of perturbed Markov decision processes. In \cite{walley2004direct}, it can be used to check the consistency of conditional probability assessments. In \cite{avrachenkov2012algorithms}, it is used to solve uniformly optimal strategies of two-player zero-sum games. Therefore, our results can be applied to all these domains as well.

\paragraph{Cylindrical Algebraic Decomposition}
To give intuition for why it is non-trivial to solve local feasibility of a 1-PLP, we now place the problem in the broader setting of a general polynomial system. An algebraic decomposition is the following process. Given a set of polynomials with $n$ variables, decompose $\mathbb{R}^n$ into regions where all polynomials have constant sign.

The local feasibility and everywhere feasibility of PLP is a special case of the above problem. To see this, we can choose the set of polynomials to be $\{a_i(\delta)x-b_i(\delta)\}$ where $\delta$ and $x$ are both considered as variables. Then the feasible region of a PLP is the projection of $\{(\delta,x):a_i(\delta)x-b_i(\delta)\ge 0,\forall i\}$ onto the space of $\delta$.

The general algorithm for computing the decomposition is cylindrical algebraic decomposition (CAD), introduced by Arnon and Collins \cite{arnon1984cylindrical} for quantifier elimination. The algorithm is powerful and solves any polynomial system, which includes local feasibility and everywhere feasibility of a general $d$-PLP. However, it is known to have doubly exponential running time: If $d$ is the degree and $m$ is the number of input polynomials, then the worst case running time of a general CAD is $dm^{2^{\Omega(n)}}$ \cite{davenport1988real,brown2007complexity,england2016complexity}. On the contrary, our algorithm in Theorem~\ref{thm:1-dim-local} has a running time polynomial in $n,d$ and $m$.

\paragraph{Robust Optimization and Robust Linear Programming}
The formulation of our problem is closely related to robust optimization, a field which aims to study optimization when parameters are uncertain \cite{bertsimas2011theory}. In robust optimization, the parameters of the optimization problem are assumed to reside in an uncertainty set, and the result should hold for any possible parameter. This formulation has broad applications in portfolio optimization, learning, and control \cite{bertsimas2011theory}.

Robust linear programming is a basic problem in robust optimization. It is written as
\begin{align*}
    \text{minimize } &c^\top  x \\
    \text{subject to } &Ax\ge b, \forall (A,b)\in \mathcal{U}
    \,.
\end{align*}
The tractability of the problem depends on the choice of $\mathcal{U}$, and the typical choices are polyhedron and ellipsoid.
If $\mathcal{U}$ is a polyhedron, the robust counterpart can be reduced to a normal LP problem~\cite{ben1999robust}. If $\mathcal{U}$ is an ellipsoid, then the robust linear program is equivalent to a second-order cone program, which can again be solved efficiently~\cite{ben1999robust}. Other choices of uncertainty sets have also been studied \cite{bertsimas2004robust,bertsimas2004price}.

To compare with PLP, we can consider the feasibility of a robust linear program (the optimization of a robust linear program can be solved by binary search and an oracle for feasibility). The feasibility problem is then to answer
\begin{align*}
    \text{whether } &\exists x \\
    \text{satisfying } &Ax\ge b, \forall (A,b)\in \mathcal{U}
    \,.
\end{align*}
PLP can be written in a similar way, as
\begin{align*}
    \forall (A,b)\in \mathcal{A}\text{, whether } &\exists x \\
    \text{satisfying } &Ax\ge b
    \,,
\end{align*}
where $\mathcal{A}$ is a $d$-dimensional algebraic curve parameterized by $\delta\in \mathbb{R}^d$, or a small enough ball on the curve. 

There are two main differences between PLP and robust linear programming. First, in PLP the feasible solution $x$ itself is not required to be robust, and can vary with $\delta$. Second, the uncertainty set is assumed to be a low-dimensional set having structure encoded by $A(\delta),b(\delta)$. In previous works, the uncertainty set is typically a geometrically much simpler set.

\subsection{Future Directions}
In this section we present some open directions. First, although the algorithm for Theorem~\ref{thm:1-dim-local} has polynomial running time, we believe the running time is far from optimal. So one possible direction is to improve the running time for local feasibility of 1-PLP. 

Second, it is discussed in the beginning that the local feasibility of PLP could be used to solve the max-flow problem where edge capacity and demand are parameterized by $\delta\in \mathbb{R}$ if we care about local behavior of $\delta$ around 0. But it is not obvious whether it is a tractable problem with $\delta$ having higher dimensions. 

One natural generalization from feasibility to optimization that is not discussed in this paper is the following problem when $\delta\in B_\epsilon(0)$: 
\begin{align*}
    \max_{x(\delta)}\ & c^\top(\delta)x(\delta)\\
    \text{subject to } & A(\delta)x(\delta)\ge b(\delta)\,.
\end{align*}
Here we want to get an optimal $x(\delta)$ as a function of $\delta$.
We know that the problem is NP-hard for dimension higher than 2 as it is harder than deciding feasibility. But for $d=1$, can the answer be expressed efficiently? If so, is there a polynomial-time algorithm that solves it? 

A conjecture in the field of probabilistic cellular automata (PCA) is the ergodicity of soldier's rule under arbitrarily small noise, discussed in Section~\ref{sec:soldiers-rule}. Our approach fails to decide the ergodicity for soldier's rule because the algorithm takes too much time on the PLP instance produced by soldier's rule. It remains open whether we can improve the running time for this specific instance of PLP and prove that the soldier's rule is ergodic under BSC noise. 

PLP is a generalization of linear programming obtained by substituting coefficients with polynomials of $\delta$. Other types of optimization problems may also admit such generalizations and it is interesting to ask whether the generalizations can be solved efficiently.


\section{Additional Results}\label{sec:additional-results}

In this section, we discuss some additional results that are closely related to local feasibility and everywhere feasibility of PLP. 

\subsection{PLP Infeasibility}

Unlike ordinary linear programming, where a program is feasible or infeasible, a PLP can be locally (everywhere) feasible, locally (everywhere) infeasible, or neither. Therefore, it also makes sense to ask whether the problems of local and everywhere \emph{infeasibility} are tractable. It turns out that feasibility and infeasibility have similar computational complexities.

\begin{theorem}\label{thm:1-dim-local-infea}
The local infeasibility of 1-PLP can be decided in polynomial time.
\end{theorem}
This is a consequence of the proof of Theorem~\ref{thm:1-dim-local}. In Section~\ref{sec:algo-dimension-1}, we will show that $B_\epsilon(0)$ can be divided into 3 parts, the negative side, the origin, and the positive side. In each part, the PLP is either always feasible or always infeasible. The algorithm solves the three parts separately, which answers whether a PLP is locally feasible, locally infeasible, or neither. 

We also have a corresponding hardness result, proved in Section~\ref{sec:hardness-infea}. 

\begin{theorem}\label{thm:hardness-infea}
The everywhere infeasibility of $d$-PLP is NP-hard for any $d$. The local infeasibility of $d$-PLP is NP-hard for any $d\ge 2$. 
\end{theorem}

\subsection{Tractable Variations of Local Feasibility}
Although the local feasibility of general $d$-PLP is NP-hard for $d\geq 2$, there are important special cases or variations where the problem becomes tractable. 

\subsubsection{The Case of Equality}
Consider restricting every constraint of the PLP to be equality, yielding
\begin{equation*}
    A(\delta)x=b(\delta).
\end{equation*}
Although it is NP-hard to decide the local feasibility of 2-PLP, it is tractable in the equality case.

\begin{theorem}\label{thm:equality}
    Let $d$ be constant. Given as input an $m\times n$ matrix $A$ and an $m$-dimensional vector $b$ of polynomials in $\delta\in \mathbb{R}^d$, there is a polynomial time algorithm answering the question: is there an $\epsilon>0$ such that for all $\delta\in B_\epsilon(0)$, there exists an $x$ satisfying $A(\delta) x = b(\delta)$?
\end{theorem}

The algorithm is given in Algorithm~\ref{alg:equality-case} in Section~\ref{sec:equality}. It partitions the space of $\delta$ into a polynomial number of regions so that the pseudoinverse of $A$ can be expressed as a rational function on each region. 

The hardness result, Theorem~\ref{thm:hardness}, immediately tells us that a similar partition cannot be done for a general PLP. In fact, in Section~\ref{sec:algo-dimension-1}, we will see that to solve the general local feasibility of PLP with a similar approach, all subsets of rows of $(A,b)$ must be considered and this leads to an exponential number of possible regions.

\subsubsection{The Case of Strict Inequality}

Another natural way to modify the local feasibility problem is to require that
all constraints be strict inequalities, yielding
$$A(\delta)x>b(\delta)\,.$$
This modification again results in local feasibility being a tractable problem.


\begin{theorem} 
        Let $d$ be constant. Given as input an $m\times n$ matrix $A$ and an $m$-dimensional vector $b$ of polynomials in $\delta\in \mathbb{R}^d$, there is a polynomial time algorithm answering the question: is there an $\epsilon>0$ such that for all $\delta\in B_\epsilon(0)$, there exists an $x$ satisfying $A(\delta) x > b(\delta)$?
\end{theorem}

The theorem is a corollary of the following observation.
\begin{lemma}
    $A(\delta)x>b(\delta)$ is locally feasible if and only if $A(0)x>b(0)$.
\end{lemma}
\begin{proof}
One direction is immediate, since local feasibility implies feasibility for all $\delta$ in a neighborhood of $0$ (and this includes $0$). Now we prove that if $A(0)x>b(0)$, then $A(\delta)x>b(\delta)$ is locally feasible. 
Fix $x$ such that $A(0)x>b(0)$.
Let $f_i(\delta)=A_i(\delta)x-b_i(\delta)$, and note that this is a polynomial in $\delta$ and hence continuous. 
Since $f_i(0)>0$, there exists $\epsilon>0$ such that $f_i(\delta)>0$ for all $\delta\in B_\epsilon(0)$. This holds for all $i$, which proves the lemma.
\end{proof}
Now, checking whether $A(0)x>b(0)$ is feasible is a linear program with strict inequality constraints. This is a special case of the subspace elimination problem discussed in Section~\ref{sec:subspace-elimination}.

\section{Applications}\label{sec:applications}

Aside from being a natural way of modeling uncertainty in linear programming, PLP can serve as a tool for solving theoretical problems. We will discuss two such problems. First, the ergodicity of elementary PCA with NAND function and BSC noise. Second, whether broadcasting of information on a 2D grid with arbitrary function and BSC noise allows recovery at infinity. We will use our algorithm for the local feasibility of 1-PLP to perform computer assisted proofs answering both questions.

\subsection{Probabilistic Cellular automata}\label{sec:application-PCA}
A probabilistic cellular automaton (PCA) is specified by a random rule for updating cells on an infinite grid. The update rule is local and homogeneous in time and space. At each time-step each site updates its state based on the states of nearby sites using the local rule. 
The deterministic version of PCA, cellular automata, was first introduced by John von Neumann and S. Ulam as a biologically motivated model for computing in the 1940s \cite{neumann1966theory}. A PCA resembles the human brain in the sense that the computation is massively parallel and local. PCA has been viewed as an important model for fault-tolerant computation \cite{toom1995cellular,gacs1986reliable}, and it also has broad applications in modeling physical and biological processes \cite{chopardcellular,ermentrout1993cellular}.

One of the main questions about a PCA is whether it is ergodic \cite{mairesse2014around,louis2018probabilistic,toom1990discrete}. A PCA is ergodic if there exists a unique stationary distribution so that starting with any initial configuration, the PCA converges to this distribution. In this section we use the methodology of polynomial linear programming developed in this paper to prove a new ergodicity result for a certain class of PCA. 


\paragraph{One-dimensional PCA.}
In this paper we consider one-dimensional discrete-time PCA. 
The PCA has alphabet $A$, so at any time $t\in \mathbb{N}$, the configuration of the PCA is denoted by $\eta_t\in A^\mathbb{Z}$
with 
$\eta_t(n)$ the value on site $n\in \mathbb{Z}$ at time $t$. The distribution of $\eta_t$ is a Markov chain as $t$ changes. Given $\eta_t$, the next configuration $\eta_{t+1}$ is obtained by updating each site $n\in \mathbb{Z}$ independently. The update rule is local, random, and homogeneous. That is, $$\eta_{t+1}(n)=f\big(\eta_t(n-a),\eta_t(n-a+1),\cdots,\eta_t(n+a)\big),$$ for some $a$ and where $f$ is a random function with a fixed distribution independent of $t$ and $n$.

A special case is \emph{elementary PCA} (see Figure~\ref{fig:elementary-PCA}). In this case the alphabet is $\{0,1\}$ and at each time step, the state at each location is set to a random Boolean function of its left neighbor and itself, i.e.,
$$\eta_{t+1}(n)=f\big(\eta_t(n-1),\eta_t(n)\big)\,.$$


\begin{figure}
    \centering
\end{figure}

\begin{figure}
\centering
     \begin{subfigure}[c]{0.4\textwidth}
         \centering
         
    \begin{tikzpicture}
    \filldraw  (0,0) circle (2pt);
    \filldraw  (-1,0) circle (2pt);
    \filldraw  (1,0) circle (2pt);
    \filldraw  (0,-1) circle (2pt);
    \filldraw  (-1,-1) circle (2pt);
    \filldraw  (1,-1) circle (2pt);
    \filldraw  (0,-2) circle (2pt);
    \filldraw  (-1,-2) circle (2pt);
    \filldraw  (1,-2) circle (2pt);
    
    \draw (-2,0)--(2,0);
    \draw (-2,-1)--(2,-1);
    \draw (-2,-2)--(2,-2);
    
    \node at (-2.5,0) {$\cdots$};
    \node at (2.5,-1) {$\cdots$};
    \node at (-2.5,-2) {$\cdots$};
    \node at (2.5,0) {$\cdots$};
    \node at (-2.5,-1) {$\cdots$};
    \node at (2.5,-2) {$\cdots$};
    
    \node at (3.5,0) {$t=0$};
    \node at (3.5,-1) {$t=1$};
    \node at (3.5,-2) {$t=2$};
    
    \node at (-1,0.5) {-1};
    \node at (0,0.5) {0};
    \node at (1,0.5) {1};
    
    \draw[-{Latex[length=6pt]}] (-1,0) -- (0,-1);
    \draw[-{Latex[length=6pt]}] (0,0) -- (1,-1);
    \draw[-{Latex[length=6pt]}] (0,0) -- (0,-1);
    \draw[-{Latex[length=6pt]}] (1,0) -- (1,-1);
    \draw[-{Latex[length=6pt]}] (-1,-1) -- (0,-2);
    \draw[-{Latex[length=6pt]}] (0,-1) -- (1,-2);
    \draw[-{Latex[length=6pt]}] (0,-1) -- (0,-2);
    \draw[-{Latex[length=6pt]}] (1,-1) -- (1,-2);
    \draw[-{Latex[length=6pt]}] (-1,0) -- (-1,-1);
    \draw[-{Latex[length=6pt]}] (-1,-1) -- (-1,-2);
    \end{tikzpicture}
    \caption{Elementary PCA }
    \label{fig:elementary-PCA}
     \end{subfigure}
     \hfill
     \begin{subfigure}[c]{0.53\textwidth}
         \centering
         
    \begin{tikzpicture}
        \draw[thick,->] (-3.3,0)--(-2.7,0);
        \draw[thick,->] (-2.3,0)--(-1.7,0);
        \draw[thick,->] (-0.7,0)--(-1.3,0);
        \draw[thick,->] (-0.3,0)--(0.3,0);
        \draw[thick,->] (0.7,0)--(1.3,0);
        \draw[thick,->] (1.7,0)--(2.3,0);
        \draw[thick,->] (3.3,0)--(2.7,0);
        
        \draw (0,0) circle (0.4);
        \draw (1,0) circle (0.4);
        \draw (3,0) circle (0.4);
    \end{tikzpicture}
    \caption{Soldier's Rule. The soldier in the middle takes a majority vote between themselves, the first person in front of them and the third person in front of them. In this case, the soldier still faces right at the next time step.}
    \label{fig:sodiers-rule}
     \end{subfigure}
     \caption{Probabilistic Cellular Automata}
\end{figure}
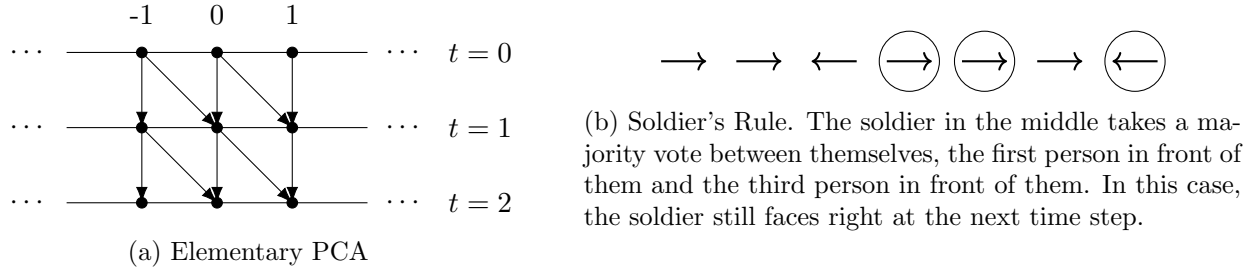

Another important example is \emph{soldier's rule}. In soldier's rule, the alphabet is $\{-1,1\}$, indicating the direction a soldier is facing. At each time step, each soldier sets its new direction based the majority vote of the direction of itself and the direction of the first neighbor and the third neighbor in its direction (see Figure~\ref{fig:sodiers-rule}). Thus,
$$\eta_{t+1}(n)=\text{Maj} \big(\eta_t(n),\eta_t(n+\eta_t(n)),\eta_t(n+3\eta_t(n))\big).$$
The probabilistic version of soldier's rule adds noise to the outcome, so each soldier has probability $p$ of being $-1$ and probability $q$ of being $1$ at each time step regardless of the configuration in the previous time step.

\paragraph{Ergodicity.}
Now we formally define ergodicity.
Let $F$ denote the transition of the Markov chain based on the update rule for $\eta_{t+1}$ in terms of $\eta_t$. 
\begin{definition}\label{def:ergodicity}
    A PCA $F$ is said to be \emph{ergodic} if there is a unique attractive invariant measure $\mu$, i.e., there exists $ \mu$ such that
    \begin{enumerate}
        \item  $\mu F=\mu$
        \item for any probability measure  $\nu$ over configurations, $\nu F^n$ converges weakly to $\mu$ as $n\to \infty$. 
    \end{enumerate}
\end{definition}

Asking whether a PCA is ergodic is similar to the question of reconstruction for broadcasting on trees (see, e.g., \cite{evans2000broadcasting}). Recall that in broadcasting on trees, the root node starts with a random state which is then propagated outward on the tree according to a given Markov transition on each edge. The reconstruction question asks whether information about the root propagates infinitely far. If we view the layers of a tree as time steps, then the reconstruction problem equivalently asks whether all initial configurations converge to the same distribution as time goes to infinity. 

Despite their similarity, unfortunately, techniques from broadcasting on trees do not directly apply to PCA. 
Given an observation at some level of the tree, the posterior at the root can be computed recursively via belief propagation, and in general the marginal distribution at any node can be expressed as a function of its children's marginals (see, e.g., \cite{bleher1995purity}). In contrast, in PCA each node is affected by multiple inputs.
Moreover, given a fixed noise model, the ergodicity of a PCA still depends on the specific form of the update rule, while broadcasting on trees depends only on the tree structure.

\paragraph{A central open problem.}
It has been an important open problem to decide the ergodicity of all symmetric elementary PCAs \cite[Chapter 7]{toom1990discrete}. It is conjectured that all non-deterministic PCAs are ergodic \cite{dobrushin1969markov}. There are three parameters for a symmetric elementary PCA: the probability to output 1 when the input is 00, 01 or 11. The conjecture was proven for most of the region of possible parameters \cite[Chapter 7]{toom1990discrete}. Specifically, regions around any deterministic functions except NAND function leads to ergodicity. Therefore, one important region of parameters where the question remained open was the corner (i.e., small noise) around the NAND function.

In \cite{holroyd2019percolation}, a potential function method was proposed for analyzing the ergodicity of PCA. The update rule they considered was a noisy NAND function with vertex noise: with probability $1-2p$ the output is the NAND of the input, and with probability $2p$ the output is a random bit. In other words,
\begin{equation}\label{eq:PCA-NAND}
    \eta_{t+1}(n)= \text{BSC}_p\big(\NAND\big(\eta_t(n-1),  \eta_t(n)\big)\big)\,,
\end{equation} 
where $\text{BSC}_p$ stands for the binary symmetric channel with crossover probability $p$.

The function in Equation~\eqref{eq:PCA-NAND} can be thought of as vertex noise, as the noise is added after the NAND function. An alternate version is that noise occurs as the information is passing from time $t$ to time $t+1$. In other words,
\begin{equation}\label{eq:PCA-NAND-edge-noise}
    \eta_{t+1}(n)= \NAND\big(\text{BSC}_p(\eta_t(n-1)),  \text{BSC}_p(\eta_t(n))\big).
\end{equation}
There is no obvious way to relate which noise (edge or vertex) leads to faster or more likely
mixing. However, it is known that mixing for PCA with vertex noise of level $p$ is implied by mixing 
of the PCA with edge noise of level $p/2$ (but a slightly different processing function),
see~Lemma 3.1 of \cite{dobrushin1977lower}. Thus, intuitively, showing ergodicity for arbitrarily
small edge-noise should be harder. 

\paragraph{Potential function method.}
In \cite{holroyd2019percolation}, the potential function essentially counts a weighted sum of the probability that certain patterns appear.
The idea is to use a coupling argument to control the distributional convergence as follows. Starting two chains with different initial conditions, one uses an auxiliary Markov chain to keep track of the locations where the two chains have coupled or have not coupled. If one can show that the density of uncoupled locations tends to zero then this implies convergence, and the potential is used to count the number of uncoupled locations by keeping track of certain patterns. 
The crucial property of a good potential function is that it is non-increasing and decreases in value over time when the chosen pattern appears.

While we know some sufficient conditions on a potential function that leads to ergodicity, the design of a good potential function can be challenging. Although the potential function method is not restricted to the specific choice of random function, the design of a potential function in \cite{holroyd2019percolation} was tailored specifically for NAND function with vertex noise. Therefore, a natural question is whether there exists a standard procedure to search for potential functions, regardless of function in the update rule of the PCA.

In this paper our approach is to encode the condition into a PLP which can then be solved on a computer to yield a good potential function. The idea of reducing the condition on such potential functions to a linear program is first discussed in \cite{makur2020broadcasting}. The key observation is that transition of the PCA is linear for a fixed noise level $p$, thus the evolution of the potential function is also linear. If the coefficients of the potential function are viewed as variables, these conditions turns out to be linear inequalities with parameter $p$, and the dependence on $p$ is polynomial.  Therefore, the feasibility of a PLP corresponds to the existence of such potential function. 
The PLP is then fed to our algorithm for the local feasibility, which certifies the existence of a potential that meets all conditions and outputs a potential that in principle can be checked by hand. 

\paragraph{Our results.}
We prove the ergodicity of PCA that uses NAND function with both vertex noise and edge noise for small enough error probability. 

\begin{theorem}
    For the PCAs described in \eqref{eq:PCA-NAND} and \eqref{eq:PCA-NAND-edge-noise}, there is an $\epsilon>0$ such that the PCA is ergodic for any $p$ on $(0,\epsilon)$.
\end{theorem}
The proof of the theorem, based on analysis of the potential function method and its connection to
feasibility of PLP, is in Section~\ref{sec:PCA} and the ($\mathbb{Q}$-rational) potentials are given
on Appendix~\ref{sec:potential-pca}.


It is worth noting that our approach of finding potential functions is not limited to PCA with two inputs, but is also a possible method to try on any 1-dimensional PCAs. In Section \ref{sec:soldiers-rule}, we will also discuss an attempt to prove ergodicity for soldier's rule and the difficulties encountered.

\subsection{Broadcasting of Information on a 2D-Grid}

Broadcasting of information on a grid is a model inspired by broadcasting on trees. In the setting of broadcasting on trees, there is an information source that generates one bit of information. The information spreads through a noisy tree with bounded degree where edges are binary symmetric channels with crossover probability $\delta$. The central question is whether the information propagates out to infinity, i.e., whether the mutual information between the root and the frontier of broadcasting converges to 0. If the information is not completely lost, then we say that \emph{reconstruction} is possible. It is shown in a series of results \cite{bleher1995purity,evans2000broadcasting} that reconstruction is possible for small enough noise if and only if $(1-2\delta)^2\text{br}(T)>1$, where $\text{br}(T)$ is the branching factor of the tree.

\paragraph{Broadcasting on the grid.} Broadcasting of information was generalized from trees to general directed acyclic graphs (DAGs) in \cite{makur2019broadcasting}, with the case of the grid being an important specialization. However, because each node is affected by multiple inputs, like PCA, the recursive approach in \cite{bleher1995purity} also does not apply to broadcasting on 2D grid. 

Let us describe the problem formally. Figure~\ref{fig:broadcasting-on-grid} illustrates how information is broadcasted on 2D grid. Since there is no edge from one quadrant to another, we can restrict attention to the non-negative quadrant.  Suppose there is one bit of information at the origin and we want to broadcast the bit to the entire infinite 2-dimensional grid. Each nodes on the grid is indexed by a pair $(t,i)$ ($t\ge 0, 0\le i\le t$). At any time $t\ge 1$, nodes with index $(t,i)$, where $0<i<t$, receives the bit from $(t-1,i-1)$ and $(t-1,i)$ and computes its own bit with a fixed Boolean function $f$. Nodes with index $(t,0)$ and $(t,t)$ receive the bit from $(t-1,0)$ and $(t-1,t-1)$, respectively, and keep what it receives as its own bit.
When a bit is passed from one node to another, it passes through a BSC channel with parameter $p$. 

\begin{figure}
     \begin{subfigure}[b]{0.45\textwidth}
    \centering
        \includegraphics[width=\textwidth]{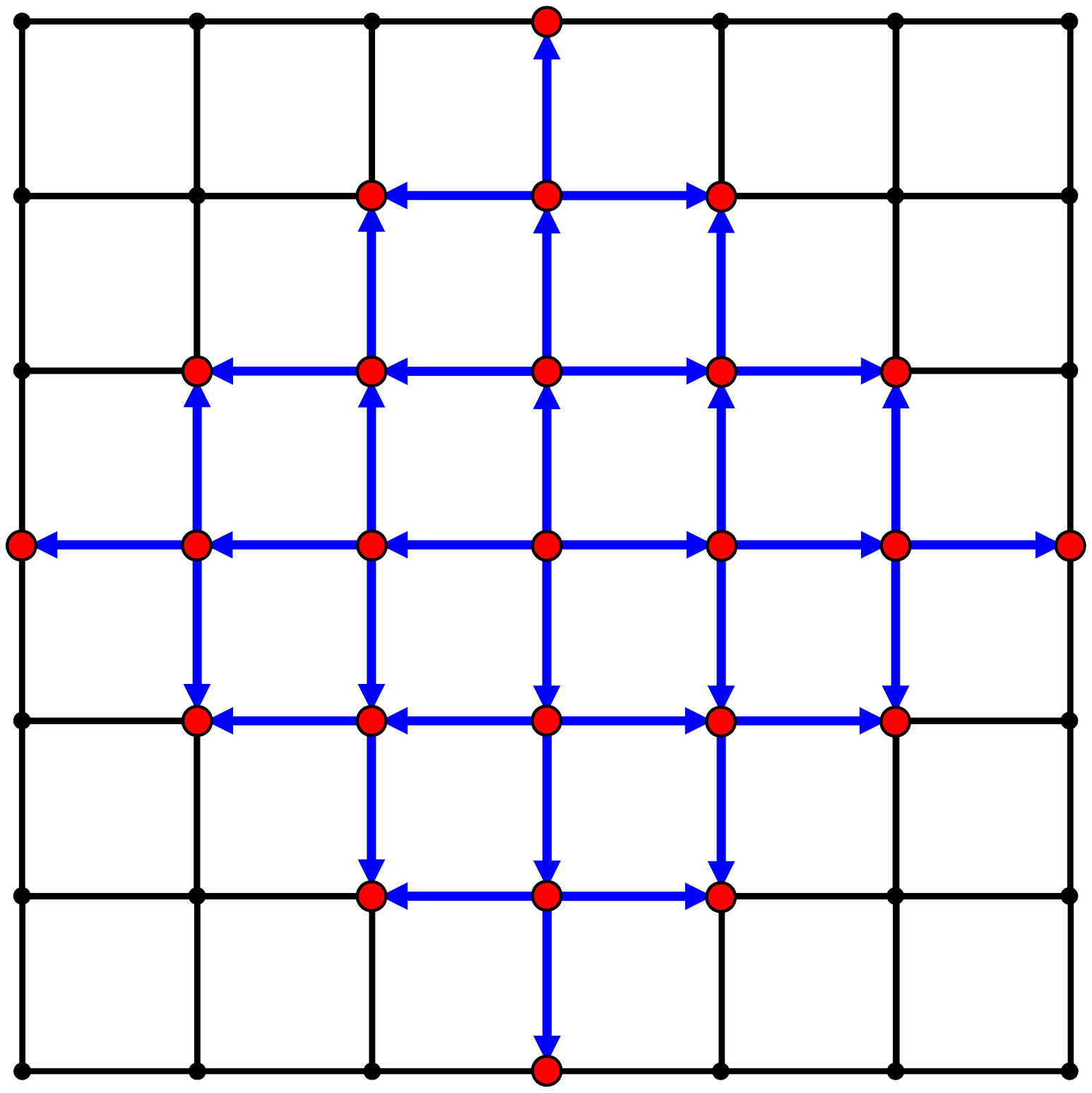}\vspace{-5mm}
    \caption{Broadcasting of Information on 2D Grid}
    \label{fig:broadcasting-on-grid}
     \end{subfigure}
     \hfill
     \begin{subfigure}[b]{0.45\textwidth}
    \centering
    \begin{tikzpicture}
    \filldraw  (0,1) circle (2pt);
    \filldraw  (-1,0) circle (2pt);
    \filldraw  (1,0) circle (2pt);
    \filldraw  (-2,-1) circle (2pt);
    \filldraw  (0,-1) circle (2pt);
    \filldraw  (2,-1) circle (2pt);

    
    \draw[-{Latex[length=6pt]}] (0,1) -- (-1,0);
    \draw[-{Latex[length=6pt]}] (0,1) -- (1,0);
    \draw[-{Latex[length=6pt]}] (-1,0) -- (-2,-1);
    \draw[-{Latex[length=6pt]}] (-1,0) -- (0,-1);
    \draw[-{Latex[length=6pt]}] (1,0) -- (0,-1);
    \draw[-{Latex[length=6pt]}] (1,0) -- (2,-1);

    \node at (0.7,1) {(0,0)};
    \node at (-0.3,0) {(1,0)};
    \node at (1.7,0) {(1,1)};
    \node at (-1.3,-1) {(2,0)};
    \node at (0.7,-1) {(2,1)};
    \node at (2.7,-1) {(2,2)};
  
    \end{tikzpicture}  \vspace{7mm}
    \caption{We can consider just one of the quadrants without loss of generality.}
    \label{fig:broadcasting-on-grid-quadrant}
     \end{subfigure}
        \caption{Broadcasting on 2D girid}
        \label{fig:broadcasting}
\end{figure}

Use $\eta_t^+$ and $\eta_t^-$ to denote the Markov chain containing bits $(t,i)$ for $0\le i\le t$ that starts with bit 1 and bit 0 respectively at the origin. We are interested in whether the information of the initial bit is lost at infinity, i.e., 
\begin{equation*}
    \lim_{t\rightarrow \infty} \TV(P_{\eta_t^+},P_{\eta_t^-})=0\,.
\end{equation*}

If we view the coordinate $t$ as time, the problem is similar to elementary PCA because each node at level $t$ receives information from two adjacent nodes from $t-1$. The only difference is that broadcasting on the grid has a bounded-length configuration, and the behavior is different at the boundary.

The potential method in \cite{holroyd2019percolation} can be extended to this problem, as discussed in \cite{makur2020broadcasting}. Instead of counting the probability of certain pattern appearing, the potential now counts the expected number of appearance of a certain pattern at time $t$. \cite{makur2020broadcasting} also showed the close relation between LP and finding the potential function. The condition on the potential is linear in the transition probabilities. Transition probabilities are polynomial in the noise parameter so the condition can be written as the feasibility of a PLP.
We will utilize our algorithm for local feasibility to find a potential function and prove the following theorem.
\begin{theorem}\label{thm:broadcasting}
    For any function $f:\{0,1\}^2\rightarrow \{0,1\}$,
    there exists $\epsilon>0$ such that for any $p\in (0,\epsilon)$, 
    $$\lim_{t\rightarrow \infty} \TV(P_{\eta_t^+},P_{\eta_t^-})=0.$$
\end{theorem}
The theorem is proved in Section~\ref{sec:broadcasting-on-grid} and the computer-discovered potentials (with entries in
$\mathbb{Q}$) are given in~\eqref{eq:pot_nand} and~\eqref{eq:pot_imp}.


One interesting open question is whether recovery is possible for higher dimensional grids. For grids with dimension 3 or more, each node has more outputs, so recovery might be easier analogously to the situation for broadcasting on trees of higher degree. Indeed, recovery is conjectured to be possible in this setting using the Majority function, in Section II.D of \cite{makur2020broadcasting}.



\section{Preliminaries}\label{sec:preliminary}
\subsection{Polynomial and Rational Functions}
We now provide several preliminary definitions and results that will be useful in the proofs.

\begin{definition}\label{def:poly}
    A degree-$D$ \emph{polynomial} in $\delta\in\mathbb{R}^{d}$ is a function with the following form
    $$p(\delta)=\sum_{\sum_{j=1}^di_j\le D,i_j\in \mathbb{Z}^+\cup\{0\} }C_{i_1,i_2\cdots,i_d}\prod_{j=1}^d\delta_j^{i_j}.$$
    Here $(C_{i_1,i_2\cdots,i_d})$ are the  coefficients of the polynomial.
    
    A degree-$D$ \emph{rational function} is a function with the following form
    $$h(\delta)=\frac{p(\delta)}{q(\delta)},$$ where $p$ and $q$ are degree-$D$ polynomials, defined on $\mathbb{R}^d\backslash\{\delta|q(\delta)=0\}$.
\end{definition}
Throughout the paper, all polynomials and rational functions are assumed to have rational coefficients. 

\begin{definition}\label{def:size-of-input}
    The \emph{size} of a polynomial or rational function is the number of bits used to describe it.
    If a polynomial or rational function is a function of $n$ variables, its degree is bounded by $D$, and the number of bits of any coefficient is bounded by $L$, then its size is $O(nDL)$.
\end{definition}

\subsection{Pseudo Inverse}
For a matrix $A\in \R^{n\times m}$, its pseudo inverse is denoted by $A^+$. 
\begin{definition}\label{def:pseudo-inverse}
    $A^+\in \R^{m\times n}$ is a pseudo inverse of $A$ if and only if 
\begin{enumerate}
    \item $AA^+A=A^+$
    \item $A^+AA^+=A$
    \item $AA^+$ and $A^+A$ are symmetric
\end{enumerate}
\end{definition}

The pseudo inverse of a matrix tells us whether a linear system is feasible.
\begin{lemma}\label{lem:pseudo_inverse}
    A linear system $Ax=b$ has a solution if and only if $AA^+b=b$, in which case $A^+b$ is a solution.
\end{lemma}
\begin{proof}
If $AA^+b=b$ holds then $A^+b$ is a solution by definition. Conversely, assume $Ax=b$ has solution $x_0$. By definition of pseudo inverse, $b=Ax_0=AA^+Ax_0=AA^+b$.
\end{proof}
The pseudo inverse can be expressed as a limit, which will be useful in our proof (see, for example \cite{albert1972regression}).
\begin{lemma}\label{lem:pseudo-inverse-limit}
    For any $A\in \mathbb{R}^{n\times m}$, its pseudo inverse $A^+$ always exists and is unique. Also,
    $$A^+=\lim_{y\rightarrow 0}(A^\top A+y I)^{-1}A^\top.$$
\end{lemma}

\subsection{Useful Lemmas in Computational Algebra}
We will use the following standard result in computational algebra.
\begin{lemma}[\cite{bugeaud2004distance}]\label{lem:dis-between-roots}
    If $x_1$ and $x_2$ are two distinct roots of a polynomial $p$ with size $s$, then $|x_1-x_2|\ge 2^{-poly(s)}$.
\end{lemma}

\begin{lemma}[\cite{coste1988thom}]\label{lem:poly-check-sign}
    Let $\delta_0$ be a real root of an irreducible  polynomial $p$ (irreducible in $\mathbb{Q}$). $\delta_0$ is uniquely determined by $p$ and signs of all derivatives of $p$ at $\delta_0$. The sign of $q(\delta_0)$ for another polynomial $q$ can be checked in time polynomial in size of $p$ and $q$.
\end{lemma}

\subsection{Notation}
For an $m\times n$ matrix $A$ we denote by $A_{-i}$ the $(m-1)\times n$ matrix obtained by removing row $i$ of $A$. For a subset $J\subseteq [m]$, the $|J|\times n$ matrix $A_J$ consists of the rows of $A$ indexed by $J$. 

The \emph{kernel} of a matrix $C\in \mathbb{R}^{m\times n}$, denoted by $\ker(C)$, consists of those $w\in \mathbb{R}^n$ such that $Cw = 0$. The \emph{row space} of matrix $C$ is denoted by $\row(C)$, which is the space spanned by all row vectors of $C$. The subspaces $\row(C)$ and $\ker(C)$ are orthogonal complements to each other by definition.




\section{A Polynomial Time Algorithm for Local Feasibility of 1-PLP}\label{sec:algo-dimension-1}
In this section we prove Theorems~\ref{thm:1-dim-local} and~\ref{thm:1-dim-local-infea} by giving a polynomial-time algorithm for local feasibility of 1-PLP. We begin in Sections~\ref{sec:rational-function-solution} and~\ref{sec:polynomial-function} by showing that although the solution $x$ can in principle be any function of the parameter $\delta$, it suffices to focus on $x$ being polynomials. Then in Section~\ref{sec:subspace-elimination} we introduce a problem called subspace elimination, reduce determining local feasibility of 1-PLP to feasibility of subspace elimination, and give an efficient algorithm for the latter problem.
Section~\ref{s:runTime} improves the runtime of the algorithm from general subspace elimination by using the structure specific to local feasibility of PLP. Section~\ref{s:OutputFeas} shows how to find a locally feasible solution for the PLP as a rational function.

We now briefly recall the problem formulation and introduce a simplification. $A$ is an $m \times n$ polynomial matrix and $b$ is an $m\times 1$ polynomial vector.
Suppose the degree of polynomials in $A$ and $b$ are bounded by $D$ with a 1-dimensional parameter $\delta$. We want to decide whether there exists $ x  (\delta)$ that satisfies the program
\begin{equation}\label{equ:lp}
    A(\delta)\cdot  x  (\delta)-b(\delta)\ge 0
    \,,
\end{equation}
for all $\delta$ in an arbitrarily 
small neighborhood $(-\epsilon,\epsilon)$ of $0$, and output $x(\delta)$ when the answer is yes. 

It turns out to be simpler, and sufficient, to derive an algorithm that decides the PLP is feasible on the positive side of the origin, i.e., whether there exists an $\epsilon$ such that the PLP is feasible in $(0,\epsilon)$. The same algorithm can be applied to the negative side. The two results can then be combined with a normal LP at the origin to decide if the PLP is feasible or infeasible in $B_\epsilon(0)$ for a small enough $\epsilon$. We will prove the following lemma and Theorems~\ref{thm:1-dim-local} and~\ref{thm:1-dim-local-infea} follow as corollaries.  

\begin{lemma}\label{lem:1-dim-local-one-side}
The problem of whether a 1-PLP is feasible on $(0,\epsilon)$ for a small enough $\epsilon$ can be solved in polynomial time.
Furthermore, if the PLP is locally feasible, a solution $x(\delta)$ can be given as a rational function of $\delta$ on the positive neighborhoods of 0. The interval $(0,\epsilon)$ on which $x(\delta)$ is a feasible solution is given by $\epsilon$ being the smallest positive root of an explicit polynomial.
\end{lemma}

\subsection{Rational Function Solution}\label{sec:rational-function-solution}
The following theorem states that feasibility implies existence of a rational function solution, and gives a constructive way of finding such a function.

\begin{theorem}\label{thm:rational_function} \footnote{This theorem can be implied by the existence of basic solution of the simplex algorithm over the ordered field of rational functions \cite{jeroslow1973asymptotic}. Here we include a direct proof for completeness.
}
    If there exists $\epsilon>0$ such that for any $\delta\in (0,\epsilon)$ there exists a $ x_1(\delta)$ such that (\ref{equ:lp}) is satisfied, then there exists a rational function $ x  (\delta)$ that satisfies (\ref{equ:lp}) for $\delta\in (0,\epsilon')$ for a positive $\epsilon'$.
    Further, we can take $ x  (\delta)$ equal to $A^+_Jb_J$ for a subset of constraints $J\subseteq[m]$, where $A_J^+$ denotes the pseudoinverse of the matrix $A_J$ consisting of the rows of $A$ indexed by $J$.
    
    The numerators and least common multiple (LCM) of denominators for all entries in $x(\delta)$ has degree at most $2nD$.
\end{theorem}

Note that $A^+_J$ is not always a rational function, but we will show in Lemma~\ref{lem:algo-rational-solution} that this is indeed the case if we choose a small enough $\epsilon$.

We begin with the observation that when a linear program is feasible, a solution can always be found by changing a subset of the inequality constraints to equalities and dropping the rest. The solution can then be expressed as a rational function of the coefficients via the pseudoinverse. 
\begin{lemma}\label{lem:equality_solution}
    For any feasible LP $L:Ax\ge b$, there exists a subset of constraints $J$ such that $\{x:A_J x=b_J\}$ is a subset of the solution set $\{x: Ax\geq b\}$. 
\end{lemma}
\begin{proof}
The proof is obtained by inductively applying the following claim. Let $S$ be a set defined by linear inequalities $Ax\ge b$ and equations $Cx=d$. We claim that either $S = \{x:Cx=d\}$, i.e. none of the inequalities are active, or there exists an inequality $a_i x\ge b_i$ such that 
$$S'= S\cap \{x:a_ix = b_i\}= \{x:A_{-i}x\ge b_{-i},Cx=d, a_ix=b_i\}
\quad \text{is nonempty}\,.$$
In words, for any feasible set of constraints we can find an inequality and change it to an equality while preserving feasibility. 

Now we prove the claim. We first suppose that for all $ w\in \ker(C)\setminus\{0\}$, $\max_{x\in S} w^\top x=+\infty$ and show that in this case $\{x:Cx=d\}=S$. Note that $\max_{x\in S}-a_i x\leq -b_i<+\infty$, so for any $i$, $a_i\not\in \ker(C)$. Because the kernel and row space are orthogonal complements, $a_i = a^0_i + a^1_i$ where $a^0_i\in \ker(C)$ and $a^1_i\in \row(C)$.
Suppose $a_i^1=\sum_{j}\alpha_jC_j$. Taking any $x$ satisfying $Cx=d$, $a^1_i x= \sum_{j}\alpha_jC_j x=
\sum_{j}\alpha_jd_j$. Hence $\max_{x\in S} -a^0_i x = \max_{x\in S} (-a_i+a^1_i)x \leq -b_i + \sum_{j}\alpha_jd_j< +\infty $ and we conclude that $a_i^0 = 0$. 
Now, since $S$ is non-empty, we may take a $y\in S$ to conclude that $b_i \leq a_iy=a_i^1 y=\sum_{j}\alpha_jd_j$ for all $i$. 
So for any $x$ in $\{x:Cx=d\}$ and any $i$, we have $a_ix=a_i^1x=\sum_{j}\alpha_jd_j\ge b_i$, which implies that $\{x:Cx=d\}=S$.

We henceforth assume that there exists $w_0\in \ker(C)$ such that $\max_{x\in S} w_0^\top x<+\infty$. Let $x^*\in S$ achieve this maximum. Note that there must exist an $i$ such that $a_i x^*= b_i$, otherwise we can increase $w_0^\top x$ by choosing $x=x^*+\epsilon w_0$ for a small enough $\epsilon>0$. So $x^*\in \{x:A_{-i}x\ge b_{-i},Cx=d, a_ix=b_i\}=:S'$, which means that $S'$ is nonempty.
\end{proof}

\begin{corollary}\label{cor:vertex_sol}
Under the condition of Theorem~\ref{thm:rational_function}, there exists $\epsilon>0$ that
for any $\delta\in (0,\epsilon)$ there exists a subset of constraints $J$ (possibly depending on $\delta$) such that $A^+_Jb_J$ is a solution.
\end{corollary}
\begin{proof}
The proof follows by combining Lemmas~\ref{lem:equality_solution} and~\ref{lem:pseudo_inverse}.
\end{proof}

Informally, Corollary~\ref{cor:vertex_sol} tells us that for fixed $\delta$, $x$ can indeed be chosen as a rational function. The conclusion can be extended to a small enough open interval, as shown in the following lemma.
    

\begin{lemma}\label{lem:algo-rational-solution}
    For a small enough $\epsilon$, it holds for every $J\subset [m]$ that $A^+_J(\delta)b_J(\delta)$ is a rational function of $\delta$ on $(0,\epsilon)$. The numerators and LCM of denominators for all entries in $A^+_J(\delta)b_J(\delta)$ has degree at most $2nD$.
\end{lemma}

\begin{proof}

By Lemma~\ref{lem:pseudo-inverse-limit}, $A_J^+=A^+_J(\delta)$ has the following limit form: 
\begin{equation*}
    A^+_J = \lim_{y\rightarrow 0}(A_J^\top A_J+y I)^{-1}A^\top_J\,.
\end{equation*}
Now, $A_J$ is a $|J|\times n$ matrix, so $A_J^\top A_J+y I$ is an $n\times n$ matrix. Each entry of $A_J^\top A_J+y I$ is a polynomial of $\delta$ and $y$ with degree at most $2D$. From the analytic formula for the inverse of a matrix, it follows that entries of $(A_J^\top A_J+y I)^{-1}$ are rational functions with denominator of degree at most $2nD$ and numerator of degree at most $2(n-1)D$. Thus, $(A_J^\top A_J+y I)^{-1}A^\top _J$ is an $n\times |J|$ matrix and the entries are rational functions of $\delta$ and $y$ with denominator and numerator of degree at most $2nD$. 

It remains to take the limit $y\to 0$.
Taking entry $i,j$ of $(A_J^\top A_J+y I)^{-1}A^\top _J$, write it as
\begin{equation*}
    f(\delta,y)=\frac{(-1)^{i+j}M_{j,i}}{\det(A_J^\top A_J+y I) }
\end{equation*}
where $M_{j,i}$ stands for the minor of matrix $A_J^\top A_J+y I$ in the $j$th row and $i$th column.
Suppose  $$(-1)^{i+j}M_{j,i}=y^{a_{ij}}(p_0^{(ij)}(\delta)+p_1^{(ij)}(\delta)y+\cdots),\qquad \text{and}$$ $$\det(A_J^\top A_J+y I) = y^a(q_0(\delta)+q_1(\delta)y+\cdots)\,.$$ We have $a_{ij}\ge a$ since the limit always exists (by virtue of the pseudoinverse existing). Therefore,
\begin{equation}\label{eq:entry-of-pseudo-inverse}
    \lim_{y\rightarrow 0}f(\delta,y)=\frac{\mathbb{I}\{a=a_{ij}\}p_0^{(ij)}(\delta)}{q_0(\delta)}
\end{equation}
for $\delta$ such that $q_0(\delta)\not=0$. $q_0(\delta)$ is not the zero function, so there exists an interval $(0,\epsilon)$ such that for $\delta\in (0,\epsilon)$, $q_0(\delta)\not=0$. We get that on $(0,\epsilon)$, $A^+_J(\delta) b_J(\delta)$ is a vector of rational functions with degree at most $2nD$, given by $$\frac{\mathbb{I}\{a=a_{ij}\}p_0^{(ij)}(\delta)}{q_0(\delta)}\,.$$
$q_0$ is a common multiple for all denominators in the vector, so the LCM for denominators of all entries has degree bounded by $2nD$.
\end{proof}

\begin{lemma}\label{lem:singleJ}
        For a small enough $\epsilon$, there exists a fixed set $J\subset [m]$ such that that $x(\delta)=A^+_J(\delta) b_J(\delta)$ is a solution of \eqref{equ:lp} for all $\delta \in (0,\epsilon)$.
\end{lemma}
\begin{proof}
Define a set of functions
\begin{equation*}
    g_J(\delta)\triangleq \begin{cases}
    1,&  \text{if } A^+_J(\delta)b_J(\delta)\text{ is a solution}\\
    0,& \text{otherwise}.\\
    \end{cases}
\end{equation*}
Now let $S(\delta)=\{J\subseteq [m]:g_J(\delta)=1\}$. From Corollary~\ref{cor:vertex_sol}, $S(\delta)$ is always non-empty for $\delta\in(0,\epsilon)$.  From Lemma~\ref{lem:number_of_change} below, $S(\delta)$ can change only finite number of times on $\mathbb{R}^+$. So we can pick $\epsilon$ small enough such that $S(\delta)$ is non-empty and remains constant on $(0,\epsilon)$. Pick one $J$ from $S(\delta)$, and $A^+_J(\delta)b_J(\delta)$ is the desired solution. 
\end{proof}

\begin{proof}[Proof of Theorem~\ref{thm:rational_function}.]
By Lemma~\ref{lem:singleJ}, we can choose a small enough $\epsilon_1$ so that there is a set $J\subset [m]$ for which $x(\delta)=A^+_J(\delta) b_J(\delta)$ is a solution for all $\delta \in (0,\epsilon_1)$. By Lemma~\ref{lem:algo-rational-solution}, there exists an $\epsilon_2$ such that $A^+_J(\delta) b_J(\delta)$ is a rational function for $\delta\in (0,\epsilon_2)$.
Now let $\epsilon'$ be the smaller of $\epsilon_1,\epsilon_2$.
\end{proof}

The following lemma was used in the proof of Lemma~\ref{lem:singleJ} above.
\begin{lemma}\label{lem:number_of_change}
    For any $J\subset [m]$, $g_J(\delta)$ has a finite number of value changes as $\delta$ varies over $\mathbb{R}$.
\end{lemma}
\begin{proof}
Note that $g_J(\delta)=1$ is equivalent to satisfying for each $i\in [m]$ the inequality 
$$a_iA^+_Jb_J-b_i\ge 0\,.$$
By Equation~\ref{eq:entry-of-pseudo-inverse}, $a_iA^+_Jb_J-b_i$ is a rational function of $\delta$ except at a finite number of points, so its sign can only change a finite number of times. When $g_J$ changes value, the sign of $a_iA^+_Jb_J-b_i$ changes for at least one $i$, which means $g_J$ has a finite number of value changes.
\end{proof}

One of the consequences of this lemma is that a PLP must be either feasible or infeasible on the positive neighborhood of 0.
\begin{corollary}\label{cor:feasible-infeasible-on-interval}
    Given a 1-PLP, we can divide $\mathbb{R}$ into a finite number of intervals such that the PLP is always feasible or infeasible when $\delta$ is within an interval.
Therefore, there exists an $\epsilon>0$ such that 
    the PLP is either feasible or infeasible on $(0,\epsilon)$. 
\end{corollary}
\begin{proof}
By Lemma~\ref{lem:number_of_change}, for each $J\in [m]$, $g_J(\delta)$ has a finite number of value changes on $\mathbb{R}$. So we can divide $\mathbb{R}$ into a finite number of intervals such that on each interval $g_J(\delta)$ is constant for all $J$. Then by Lemma~\ref{lem:equality_solution}, whether the PLP is feasible is also fixed. 
\end{proof}

\begin{proof}[Proof of Lemma~\ref{lem:testing-small-delta}]

As discussed in Lemma~\ref{lem:singleJ} and Lemma~\ref{lem:number_of_change}, if the sign of 
\[g_{J,i}(\delta) = a_iA_J^+b_J-b_i\]
does not change on $(0,2\delta_0)$ for any $J\subset [m]$, $|J|\le n$ and $i\in [m]$, then the feasibility of the PLP remains the same on the interval. 

We can rewrite the function as 
\[\lim_{y\rightarrow 0}a_i(A_J^\top A_J+yI)^{-1}A_J^\top b_J-b_i.\]
In the limit we have a rational function of $y$ where coefficients are polynomials of $\delta$. If the sign of all such polynomials remains constant, the sign of $g_{J,i}$ remains constant. Each polynomial has rational coefficients that has number bounded by $n! \cdot (2D 2^L)^{2n}\le 2^{4n(L+\log Dn)}$.
By Cauchy's bound, the minimum root of such polynomial is at least $2^{-8n(L+\log Dn)}$. We choose this value as $2\delta_0$. So the feasibility of the PLP remains constant for $(0,2\delta_0)$
\end{proof}

\begin{remark}
The number of bits required for $\delta_0$ in Lemma~\ref{lem:testing-small-delta} is tight up to logarithmic factors. Consider the following 1-PLP,
\begin{align*}
    & 0\le x_1\le 2^{-L}\\
    & 0\le x_i\le 2^{-L}x_{i-1},\forall 2\le i\le n\\
    & x_n\ge \delta\,.
\end{align*}
The PLP is feasible if and only if $0\le \delta\le 2^{-nL}$. So $\delta_0$ chosen in Lemma~\ref{lem:testing-small-delta} is at most $2^{-nL}$.
\end{remark}


\subsection{Polynomial Solution and Reformulation of Feasibility Problem}\label{sec:polynomial-function}
In this subsection we show that when the 1-PLP is locally feasible, it has a polynomial solution $x(\delta)$. We then use this to reformulate the feasibility problem in terms of a condition on the coefficients of the polynomials. 

\paragraph{Polynomial Solution}
Starting from the conclusion of Theorem~\ref{thm:rational_function}, let $ x  (\delta)$ be the rational solution of (\ref{equ:lp}). Suppose the LCM of all denominators of $ x  _i(\delta)$ is $q(\delta)$ and $p(\delta)= q(\delta)\cdot x(\delta) $. The degree of $q$ is bounded by $2nD$. 


Assume $q(\delta)=\delta^c(1+\delta q_1(\delta))$ where $0\le c\le 2nD$. We can rewrite (\ref{equ:lp}) in the form
\begin{equation}\label{equ:poly_LP}
    A(\delta)\cdot p(\delta)-\delta^{c+1}\cdot q_1(\delta)\cdot b(\delta)-\delta^c\cdot b(\delta)\ge 0
    \,,
\end{equation}
 By Theorem~\ref{thm:rational_function}, it is a vector of polynomials of degree bounded by $2nD$. Now we have a new instance of the problem which is guaranteed to have a polynomial solution (for at least one $c$). So we can solve the problem in the space of the coefficients of $p(\delta)$ and $q_1(\delta)$. The total number of variables then becomes $O(n^2D)$.

\paragraph{Reformulation of Local Feasibility}
We will solve the local feasibility of Equation~\ref{equ:poly_LP} for all possible $c$ in $0\le c\le 2nD$, the original problem would be feasible if and only if one of the instances is feasible. By letting 
$$x'(\delta)=\big(p(\delta),q_1(\delta)\big)^\top$$
and
$$A'(\delta)=
\begin{pmatrix}
A(\delta) & 0\\
0& -\delta^{c+1}b(\delta)\\
\end{pmatrix},$$
Equation~\ref{equ:poly_LP} can be simplified to 
\begin{equation}\label{eq:poly_LP-simplified}
    A'(\delta)x'(\delta)- \delta^c\cdot b(\delta)\ge 0
\end{equation}
Let 
$$x'_i(\delta)=\sum_{j=0}^{2nD}h_{ij}\delta^j,$$
the goal is to decide feasibility in the possible space of $(h_{ij})_{ij}$.

Let $f_i(\delta)=A'_i(\delta)\cdot x'(\delta)-\delta^cb_i(\delta)$, and $[f_i]_j$ be the coefficient of $\delta^j$ in $f_i$. $[f_i]_j$ is a linear function of $(h_{kl})_{kl}$. Then $f_i\ge 0$ around $\delta=0$ is equivalent to that the first non-zero coefficient of $f_i$ is larger than 0 at 0 or $f_i=0$. This translates to the OR of some linear constraints over $(h_{kl})_{kl}$. Thus, considering all the constraints, we get that the PLP is locally feasible if and only if 
\begin{align}\label{eq:and-of-or}
\begin{split}
   \text{For every } i\in [m]:\qquad &[f_i]_0>0\\
    &\text{ OR }[f_i]_0=0, [f_i]_1>0\\
    &\text{ OR }[f_i]_0=0, [f_i]_1=0, [f_i]_2>0 \\
    &\;\;\;\;\vdots\\
    & \text{ OR }[f_i]_0=0, [f_i]_1=0,\cdots, [f_i]_{d-1}=0, [f_i]_d'\ge 0,
\end{split}
\end{align}
where $D'$ is the total degree of $f_i(\delta)$, $D'\le 2nD$.

The algorithmic task is reduced to the following: For all $c$ between $0$ and $2nD$, check whether there exists $(h_{kl})_{kl}$ satisfying \eqref{eq:and-of-or}.





\subsection{Subspace Elimination}\label{sec:subspace-elimination}
In this section, we will reduce the problem in Section~\ref{sec:polynomial-function} to a new problem called subspace elimination and solve the problem in polynomial time.

\subsubsection{Definition of Subspace Elimination}

\begin{definition}\label{def:subspace-elimination}
    In the \emph{subspace elimination} problem, the input is a collection $B_1,\dots, B_m$ of non-empty open half-spaces of subspaces in $\mathbb{R}^n$, where each $B_i$ is defined by the following equalities and inequality:
    \begin{align*}
        a_{i,1} x &= b_{i,1} \\
        a_{i,2} x &= b_{i,2} \\
        &\;\;\vdots\\
        a_{i,k_i}x&=b_{i,k_i}\\
        c_ix&<d_i
        \,.
    \end{align*}
    The goal is to decide whether or not
    $$\mathbb{R}^n\backslash \left(\bigcup_{i=1}^m B_i\right)$$
    is empty. 
\end{definition}

Note that this definition is a generalization of LP feasibility. It is also a generalization of the problem we have in Section~\ref{sec:polynomial-function}. 
\begin{lemma}
    The feasibility of \eqref{eq:and-of-or} is a special case of the subspace elimination problem.
\end{lemma}
\begin{proof}
To see this, let us take the negation of \eqref{eq:and-of-or}, which becomes
\begin{align}\label{eq:or-of-or}
\begin{split}
    \exists i,\ &[f_i]_0<0\\
    &\text{ OR }[f_i]_0=0, [f_i]_1<0\\
    &\text{ OR }[f_i]_0=0, [f_i]_1=0, [f_i]_2<0 \\
    &\;\;\;\;\vdots\\
    & \text{ OR }[f_i]_0=0, [f_i]_1=0,\cdots, [f_i]_{d-1}=0, [f_i]_d'< 0
    \,.
\end{split}
\end{align}
One can check that this is indeed the negation of \eqref{eq:and-of-or} by taking the union of corresponding terms. Then the negation of the AND of \eqref{eq:and-of-or} for all $i$ becomes the OR of \eqref{eq:or-of-or} for all $i$, which coincides with the definition of subspace elimination.
\end{proof}

Therefore, it is sufficient to find a polynomial-time algorithm for subspace elimination.

\subsubsection{Algorithm for Subspace Elimination}
\begin{theorem}\label{thm:sub-elimination}
The subspace elimination problem defined in Definition~\ref{def:subspace-elimination} can be solved by Algorithm~\ref{alg:subspace_elimination} with running time $O(m^2K\cdot \mathsf{LP}(n,mK,L))$, where $K=\max_ik_i$ is the maximum number of equalities among all subspaces, $\mathsf{LP}(a,b,c)$ stands for the running time of ordinary linear programming with $n$ variables, $m$ constraints and $L$-bit numbers
\end{theorem}

Let $\aff(A)$ denote the affine hull of a set $A$. The algorithm for subspace elimination is Algorithm~\ref{alg:subspace_elimination}.
\begin{algorithm}
\caption{Subspace Elimination}\label{alg:subspace_elimination}
\begin{algorithmic}[1]
\State \textbf{input:} $\mathcal{B}=\{B_i:i\in [m]\}$
\State $P\leftarrow \mathbb{R}^n$, $\mathcal{A}\leftarrow \mathcal{B}$
\While{$P\not=\emptyset$ and $\exists B\in \mathcal{A},P\subseteq \aff(B)$}
    \State Find $B\in \mathcal{A}$ such that $P\subseteq \aff(B)$
    \State $P\leftarrow P\backslash B$
    \State $\mathcal{A}\leftarrow \mathcal{A}\backslash \{B\}$
\EndWhile
\If{$P=\emptyset$}
\State \textbf{return} False
\Else 
\State \textbf{return} True
\EndIf
\end{algorithmic}
\end{algorithm}

Let us explain some details of the algorithm.  During the process, $P$ is always represented by a linear program, i.e., a set of linear constraints.  First, we start with  $P$ having no constraints. Whenever we change $P$ in line 5, the chosen $B_i$ in line 3 satisfies $P\subseteq \aff(B_i)$. So $P\backslash B_i$ is simply $P$ with a new constraint $c_ix\ge d_i$.  
To check whether $P\subset \aff(B_i)$ in line 3, notice that 
$\aff(B_i)$ is the set of all $x$ satisfying 
\begin{align*}
    a_{i,1} x &= b_{i,1} \\
    a_{i,2} x &= b_{i,2} \\
    &\;\;\vdots\\
    a_{i,k_i}x&=b_{i,k_i}.
\end{align*} 
We can check whether $P\subset \aff(B_i)$ by solving whether $$\max_{x\in P}a_{i,j} x=\min_{x\in P} a_{i,j} x = b_{i,j}$$ for all $1\le j\le k_i$.
The operation of deciding whether $P$ is empty is given by solving a normal LP.

\begin{proof}[Proof of Theroem~\ref{thm:sub-elimination}]
To prove the correctness of the algorithm, We need to show that $P\backslash \left(\bigcup_{D\in \mathcal{A}}D\right)\not=\emptyset$ if $P\not\subseteq \aff(D)$ for any $D\in \mathcal{A}$. This means that $\aff(P)\not\subseteq \aff(D)$. Observe that 
$$\text{dim}(P\cap D)\le \text{dim}(\aff(P)\cap \aff(D))<\text{dim}(\aff(P))=\text{dim}(P)$$ if $\aff(P)\not\subseteq \aff(D).$ So
$$P\backslash \left(\bigcup_{D\in \mathcal{A}}D\right) = P\backslash \left(\bigcup_{D\in \mathcal{A}}(P\cap D)\right)$$
cannot be empty.

Let us analyze the running time for subspace elimination. 
In the worst case, Algorithm~\ref{alg:subspace_elimination} checks whether $P\subset B_i$ for some $i$ for $O(m^2)$ times. During each check, it uses $O(K)$ times of LP with $n$ variables and at most $mK$ constraints. Therefore, the total running time is $O(m^2K\cdot \mathsf{LP}(n,mK,L))$ where $\mathsf{LP}(a,b,L)$ stands for the running time of linear programming with $a$ variables, $b$ constraints and input numbers bounded by $L$ bits.\end{proof}

\subsection{Improvement in Running Time}\label{s:runTime}

Let us derive the running time for local feasibility of PLP using the time bound from the last section. Suppose the input has $n$ variables, $m$ constraints with degree bounded by $D$, and the input numbers are bounded by $L$ bits. In this section, the parameters are for the input PLP, not the input of subspace elimination, which is different from Section~\ref{sec:subspace-elimination}. The input dimension for subspace elimination is $O(n^2D)$, the number of subspaces is $O(mnD)$, and $K$ is $O(nD)$. We need to call subspace elimination at most $2nD$ times for different choices of $c$ in Section~\ref{sec:polynomial-function}. So the final running time is $O\big(n^4m^2D^4\cdot \mathsf{LP}(O(n^2D),O(n^2mD^2),L)\big)$. (Recall that $\mathsf{LP}(a,b,L)$ stands for the running time of linear programming with $a$ variables, $b$ constraints and input numbers bounded by $L$ bits.)

However, local feasibility of PLP satisfies special structure that allows us to further reduce the running time. The key observation is that for the subspaces that belong to the same constraint, i.e., subspaces in \eqref{eq:or-of-or} forms a chain of embeded affine hull.
Specifically, let us denote the set on the space of
\begin{equation*}
    [f_i]_0=0, [f_i]_1=0, [f_i]_2=0,\cdots [f_i]_j=0, \cdots[f_i]_{j+1}<0
\end{equation*}
by $B_{i,j}$, where $0\le j\le d'-1$.
Then we have the following lemma.
\begin{lemma}\label{lem:embeded-affine-space}
We have the sequence of inclusions
\begin{equation*}
\begin{aligned}
    \mathbb{R}^n=\aff(B_{i,0})\supset \aff(B_{i,1}) \supset \cdots \supset \aff(B_{i,D'-1})\,.
\end{aligned}
\end{equation*} 
\end{lemma}
Recall that $D'=O(nD)$ is the bound on the degree of $x(\delta)$.

Thus, if $P\not\subset \aff(B_{i,j})$, then $P\not\subset \aff(B_{i,j'})$ for any $j'>j$. This means that when we are looking for $D_{i,j}$ such that $P\subset\aff(B_{i,j})$, we should keep track of 
$j^*_i\defeq \max_j\{j:P\subset\aff(B_{i,j})\}$ and look at $B_{i,j^*_i+1}$. If $P\not\subset \aff(B_{i,j})$, there is no need to look at $B_{i,j}$ with bigger index $j$.

We can further simplify Algorithm~\ref{alg:subspace_elimination} to Algorithm~\ref{alg:subspace-elimination-for-PLP}. 
\begin{algorithm}
\caption{Subspace Elimination for Local Feasibility of PLP}\label{alg:subspace-elimination-for-PLP}
\begin{algorithmic}[1]
\State \textbf{input:} $\mathcal{B}=\{B_{i,j}:i\in [m],0\le j\le D'-1\}$ with form
\State $P\leftarrow \mathbb{R}^n$, $\mathcal{A}\leftarrow \mathcal{B}$
\State $\forall i\in [m]$, $j^*_i\leftarrow -1$
\While{$HasUpdate$}
    \State $i\leftarrow 1$, $HasUpdate\leftarrow False$
    \While{$i\le m$}
        \If{$P\subseteq\aff(B_{i,j^*_i+1})$}
            \State $P\leftarrow P\backslash B_{i,j^*_i+1}$
            \State $\mathcal{A}\leftarrow \mathcal{A}\backslash\{B_{i,j^*_i+1}\}$
            \State $j^*_i\leftarrow j^*_i+1$
            $HasUpdate\leftarrow True$
        \EndIf
        \State $i\leftarrow i+1$
    \EndWhile
\EndWhile
\If{$P=\emptyset$}
\State \textbf{return} False
\Else 
\State \textbf{return} True with a point in the relative interior of $P$ as feasible solution
\EndIf
\end{algorithmic}
\end{algorithm}

When the algorithm stops, we have either $P=\emptyset$ or for any $i$, $P\not\subset \aff(D_{i,j^*_i+1})$. Thus $P\not\subset \aff(D_{i,j})$ for any $j>j^*_i$. In other words, for any $D\in \mathcal{A}$, we have $P\not\subset \aff(D)$. So Algorithm~\ref{alg:subspace-elimination-for-PLP} is correct. 

Compared to directly using Algorithm~\ref{alg:subspace_elimination}, we only test whether $P\subset\aff(D_{i,j})$ at most once for each $i,j$ rather than potentially $O(nmD)$ times. For each $i,j$ pair, we solve $O(D')=O(nD)$  linear programming instances. There are $O(mD')=O(nmD)$ $i,j$ pairs, so the number of calls to the original LP is $O(nD\cdot mD')=O(mn^2D^2)$. Also, the number of constraints for each original LP is at most $O(mD')$. This means the running time for Algorithm~\ref{alg:subspace-elimination-for-PLP} is $O(n^2mD^2\cdot \mathsf{LP}(O(n^2D),O(nmD),L))$. Recall in Section~\ref{sec:polynomial-function}, to solve local feasibility of 1-PLP, Algorithm~\ref{alg:subspace-elimination-for-PLP} needs to run $O(nD)$ times for different $c$, so the final running time is $O(n^3mD^3\cdot  \mathsf{LP}(O(n^2d),O(nmD),L))$.

\subsection{Outputting a Feasible Solution and its Range}
\label{s:OutputFeas}
We promised in Theorem~\ref{thm:1-dim-local} and Lemma~\ref{lem:1-dim-local-one-side} to give a feasible solution. Note that in Section~\ref{sec:subspace-elimination}, the algorithm only returns whether the instance is feasible, but does not give a feasible solution when the answer is yes.
The special structure of subspace elimination \eqref{eq:or-of-or} also implies the following property which allows us to output a feasible solution with Algorithm~\ref{alg:subspace-elimination-for-PLP}. 
\begin{lemma}
    When Algorithm~\ref{alg:subspace-elimination-for-PLP} returns true, 
    \begin{align*}
        \mathrm{Relint} (P) &= \{(h_{ij})_{ij}|\forall i,\ [f_i]_0=0, [f_i]_1=0, \cdots, [f_i]_{j^*_i-1}=0, [f_i]_{j^*_i}>0\} \\
        &\subset \mathbb{R}^n\backslash \left( \cup_{i\in [m],0\le j\le D'-1} B_{i,j} \right)
    \end{align*}
    where $\mathrm{Relint}$ stands for the relative interior of a set. 
\end{lemma}
\begin{proof}
We will prove this by showing $\text{Relint } (P)\cap B_{i,j}=\emptyset$ for any $i,j$. 
First, when the algorithm ends,
\begin{equation*}
   P=\mathbb{R}^n\backslash \left( \cup_{i\in [m],0\le j\le j^*_i} B_{i,j} \right) .
\end{equation*}
So $\text{Relint } (P)\cap B_{i,j}=\emptyset$ for any $i\in [m]$ and $0\le j\le j^*_i$.

Also, $P\subset \aff(B_{i,j^*_i})$ for any $i$ by the algorithm.
So based on \eqref{eq:or-of-or}, $P$ can be explicitly expressed as the following set.
\begin{equation*}
    \forall i,\ [f_i]_0=0, [f_i]_1=0, \cdots, [f_i]_{j^*_i-1}=0, [f_i]_{j^*_i}\ge  0
\end{equation*}

Note that when the algorithm terminates, $P\not\subset \{[f_i]_{j^*_i}=0\}$. So the set $\{[f_i]_{j^*_i}=0\}\cap P$ is on the relative boundary of $P$.
Therefore, $\aff(P)$ is given by
\begin{equation*}
    \forall i,\ [f_i]_0=0, [f_i]_1=0, \cdots, [f_i]_{j^*_i-1}=0.
\end{equation*}
and $\text{Relint } (P)$ is the following set.
\begin{equation*}
    \forall i,\ [f_i]_0=0, [f_i]_1=0, \cdots, [f_i]_{j^*_i-1}=0, [f_i]_{j^*_i}>  0.
\end{equation*}
By comparing the set with \eqref{eq:or-of-or}, we know that for any $i\in [m]$ and $j>j^*_i$, $B_{i,j}\cap \text{Relint } (P)=\emptyset$.
\end{proof}

Finding a point in the relative interior of a polyhedron is a standard initialization step for interior point methods. For a concrete algorithm see for example \cite{cartis2006finding}. The feasible point for \eqref{eq:and-of-or} can be used to generate a feasible solution for the original local feasibility of PLP problem using $x(\delta)=\frac{p(\delta)}{\delta^c(1+q_1(\delta))}$.

So far we have generated a rational function solution $x(\delta)=\frac{p(\delta)}{q(\delta)}$ that satisfies the PLP on range $(0,\epsilon)$ for a small enough $\epsilon$ (without loss of generality, $q(\delta)>0$ when $\delta$ is a small enough positive number). $x(\delta)$ is feasible if and only if it satisfies $A_i(\delta)p(\delta)-b_i(\delta)q(\delta)\ge 0$ and $q(\delta)>0$ for any $i$. The largest possible $\epsilon$ such that $x$ is feasible on $(0,\epsilon)$ is given by the smallest positive root among polynomials $A_i(\delta)\cdot p(\delta)-b_i(\delta)q(\delta)$ and $q(\delta)$.


\section{Hardness of Polynomial Linear Programming}\label{sec:hardness}
In this section we prove Theorem~\ref{thm:hardness} and Theorem~\ref{thm:hardness-infea}.

\subsection{Local Feasibility to Everywhere Feasibility}
In this section, we will show that everywhere feasibility (infeasibility) of 1-PLP can be reduced to the local feasibility (infeasibility) of 2-PLP. Note that feasibility and infeasibility are not complements of one another, so we have separate but similar arguments for feasibility and infeasibility. Furthermore, show that everywhere feasibility of 1-PLP does is not easier even if we restrict the polynomials in $A(\delta)$ and $b(\delta)$ to be linear functions.

\begin{lemma}\label{lem:localPLP2everywherePLP}
    There is a polynomial-time reduction from the everywhere feasibility (infeasibility) of 1-PLP to the local feasibility (infeasibility) of 2-PLP.
\end{lemma}
\begin{proof}
Let us first show the reduction on feasibility. For any instance of 1-PLP 
\begin{equation}\label{eq:1-PLP}
    (A(\delta), b(\delta)),
\end{equation}
let $D$ be the maximum degree of $A(\delta)$ and $b(\delta)$. For the reduction of feasibility, consider an instance of 2-PLP,
\begin{equation}\label{eq:2-PLP}
    \left( \delta_2^{D+1}A \left(\frac{\delta_1}{\delta_2}\right), \delta_2^{D+1} b\left(\frac{\delta_1}{\delta_2}\right)\right).
\end{equation}
We claim that \eqref{eq:1-PLP} is everywhere feasible if and only if \eqref{eq:2-PLP} is locally feasible. 

If \eqref{eq:1-PLP} is infeasible at $\delta=a$, then \eqref{eq:2-PLP} is infeasible when $\delta_1/\delta_2=a$ since they only differ by a multiplication of constant on both sides. So \eqref{eq:2-PLP} cannot be locally feasible.

Conversely, if \eqref{eq:1-PLP} is feasible for all $\delta\in \mathbb{R}$, then similarly \eqref{eq:2-PLP} is feasible for any $\delta_1/\delta_2\in \mathbb{R}$. When $\delta_2=0$, \eqref{eq:2-PLP} becomes $(0,0)$, which is feasible. So \eqref{eq:2-PLP} is feasible for any $(\delta_1,\delta_2)\in\mathbb{R}^2$.

To show the reduction for infeasibility, let us consider the following instance of 2-PLP with $n+1$ variables and $m+1$ constraints
\begin{equation}\label{eq:2-PLP-infea}
    \left\{
    \begin{aligned}
    \delta_2^{D}A \left(\frac{\delta_1}{\delta_2}\right)x\ge  \delta_2^{D} b\left(\frac{\delta_1}{\delta_2}\right)\\
    x'\delta_2=1
    \end{aligned}
    \right.
\end{equation}
where $x'$ is a variable different from $x$. We claim that \eqref{eq:1-PLP} is everywhere infeasible if and only if  \eqref{eq:2-PLP-infea} is locally infeasible. The argument is similar as feasibility. If \eqref{eq:1-PLP} is feasible at $\delta=a$, then \eqref{eq:2-PLP} is feasible when $\delta_1/\delta_2=a$. 

If \eqref{eq:1-PLP} is infeasible for any $\delta\in \mathbb{R}$, then \eqref{eq:2-PLP-infea} is infeasible for any $\delta_1/\delta_2\in \mathbb{R}$. And \eqref{eq:2-PLP-infea} is always infeasible when $\delta_2=0$ since the last equality cannot be satisfied. So \eqref{eq:2-PLP-infea} is locally infeasible.
\end{proof}

Now we show that everywhere feasibility is equally as hard as the degree-1 case.
\begin{definition}[Degree-1 PLP]
    A \emph{degree-1 PLP} problem is a PLP problem where $A(\delta)$ and $b(\delta)$ are linear functions of $\delta$.
\end{definition}

\begin{lemma}
There is a polynomial-time reduction from the everywhere feasibility of 1-PLP to the everywhere feasibility of degree-1 1-PLP.
\end{lemma}

\begin{proof}
For a PLP instance $A(\delta)$ and $b(\delta)$, suppose the maximum degree is $d$ and write
\begin{equation*}
    A_{ij}=\sum_{k=0}^d a_{ijk}\delta^k,
\qquad
    b_i = \sum_{k=0}^d b_{ik}\delta^k.
\end{equation*}
We design the following degree-1 PLP instance so that it is equivalent to the original one. There are $n(d+1)+d$ variables, $x_i^{(j)}$ for $i\in [n]$, $j=[d]\cup\{0\}$ and $y^i$ for $i\in [d]$. First, let $$x_i^{(j)}=\delta x_i^{(j-1)}$$ for any $i\in [n]$ and $j\in [d]$, so $x_i^{(j)}$ corresponds to $x_i\delta^j$ in the original PLP instance. Next, let  $$y_1=\delta,\ y_i=\delta y_{i-1}$$ for $i\ge 2$, so $y_i=\delta^i$. Finally, we add $$\sum_{j=1}^n\sum_{k=0}^d a_{ijk}x_i^{(k)}\ge \sum_{k=0}^d b_{ik}y_i$$ for any $i\in [m]$. It is straightforward to check that by construction, the feasibility of the degree-1 PLP instance is the same as the original PLP instance, i.e., each feasible point in one problem corresponds to a feasible point in the other.
\end{proof}

Next we show that the everywhere feasibility (infeasibility) of 1-PLP problem is NP-hard, which implies that both the everywhere feasibility of degree-1 1-PLP problem and the local feasibility (infeasibility) of 2-PLP problem are also NP-hard.

\subsection{Hardness of Everywhere Feasibility of PLP}
In this section we only focus on 1-PLP, so any PLP instance is 1-PLP by default. 

\begin{theorem}
    The everywhere feasibility of 1-PLP problem is NP-hard.
\end{theorem}

The proof has 3 steps. In the first step, we take the dual and negation of the original PLP to get a form of the problem that is easy to analyze. In the second step, we reduce the independent set problem to this form of PLP by considering a quadratic instance of the form. In the third step, we show that the instance has size bounded by polynomial.

\paragraph{Step 1: Dual and negation.}
We will first take the dual and negation of the program, then use a construction from \cite{pardalos1991quadratic}.
From linear programming duality, 
\begin{equation*}
    \forall\delta, \text{ the program } A(\delta)x\ge b(\delta) \text{ is feasible } 
\end{equation*}
is equivalent to 
$$\forall \delta, y,\text{ s.t. } y\ge 0 \text{ and } A^\top (\delta)y=0, \text{ it holds that } b^\top (\delta)y\le 0\,.$$
Negation of this is 
\begin{equation}\label{equ:exists}
    \exists \delta,y\text{ satisfying }y\ge 0 \text{ and } A^\top (\delta)y=0\text{ with }b^\top (\delta)y> 0\,.
\end{equation}
\eqref{equ:exists} is easier to analyze than the original form of problem because it asks the existence of a solution, while the original problem asks to verify a property for arbitrary $\delta$.

\begin{lemma}\label{lem:LPtransform}
    Fix $A\in \mathbb{R}^{m\times n}$ and $c\in \mathbb{R}^n$ such that there is no $x^*\in \mathbb{R}^n$ simultaneously satisfying
    \begin{equation}\label{equ:degenerate-form}
    \left\{
    \begin{aligned}
        c^\top x^*>0\\
        Ax^*\ge 0\,.
    \end{aligned}
    \right. 
    \end{equation}
    Then for any $b\in \mathbb{R}^m$ and  $d\in \mathbb{R}$ there exists a solution $x\in \mathbb{R}^n$ of  
    \begin{equation}\label{equ:first-form}
    \left\{
    \begin{aligned}
        c^\top x>d\\
        Ax\ge b
    \end{aligned}
    \right. 
    \end{equation}
    if and only if there exists a solution $x_1,x_2\in \mathbb{R}^n$, $w\in \mathbb{R}$, $s\in \mathbb{R}^m$ of
    \begin{equation}\label{equ:second-form}
    \left\{
    \begin{aligned}
        c^\top (x_1-x_2)-dw>0\\
        A(x_1-x_2)=bw+s\\
        x_1,x_2,w,s\ge 0\,.
    \end{aligned}
    \right. 
    \end{equation}
\end{lemma}
\begin{proof}
Suppose program (\ref{equ:first-form}) has a feasible solution $x$. Then letting $x_1=\max\{x,0\}$, $x_2=\max\{-x,0\}$, $w=1$, $s=Ax-b$, we get a feasible point for program (\ref{equ:second-form}).

Conversely, suppose program (\ref{equ:second-form}) has a feasible solution $x_1,x_2,w,s$. We claim that $w\not=0$. Otherwise, by taking $x^*=x_1-x_2$, we get $c^\top x^*>0$ and $Ax^*=s\ge 0$, i.e., a feasible solution for (\ref{equ:degenerate-form}) -- a contradiction. So we can assume $w\neq 0$ and let $x=(x_1-x_2)/w$ so that $c^\top x>d$ and $Ax-b=s/w\ge 0$. This means program (\ref{equ:first-form}) is feasible.
\end{proof}

Note that \eqref{equ:second-form} is a special case of \eqref{equ:exists}, we have the following corollary.
\begin{corollary}
   Suppose there is a polynomial time algorithm for testing everywhere feasibility of any PLP. Then there is a poly-time algorithm for solving the following problem. Given $m\times n$ matrix polynomial $A(\delta)$, $m$-dimensional vector polynomial $b(\delta)$, $n$-dimensional vector polynomial $c(\delta)$ and polynomial $d(\delta)$, decide whether 
   \begin{equation}\label{eq:lp2}
          \exists x, \delta \quad \left\{
       \begin{aligned}
       &c^\top (\delta)x>d(\delta)\\
       &A(\delta)x\ge b(\delta)
       \end{aligned}
       \right. 
   \end{equation}
   where $A(\delta),c(\delta)$ satisfies for any $\delta\in \mathbb{R}$, the following program has no solution:
    \begin{equation} \label{eq:nondeg_cond}
    \left\{
    \begin{aligned}
        c^\top (\delta)x^*&>0\\
        A(\delta)x^*&\ge 0\,. 
    \end{aligned}
    \right.
    \end{equation} 
\end{corollary}

\paragraph{Step 2: Independent set to PLP.}
Our next step is to take any graph $G$ and a produce an instance $A(\delta),b(\delta), c(\delta), d(\delta)$ satisfying~\eqref{eq:nondeg_cond}, such that~\eqref{eq:lp2} is feasible if and only if there exists an independent set of size $k$ in~$G$. 


\begin{lemma}[\cite{pardalos1991quadratic}]\label{lem:PV91}
    There exists a polynomial-size linear program $P$,
    \begin{equation}\label{equ:LP-PV91}\tag{$P$}
    \begin{aligned}
        &0\le x_i\le 1,1\le i\le n\\
        &y_{ij}\ge 0,y_{ij}\ge x_i+x_j-1,1\le i<j\le n\\
        &\delta = \sum_{i=1}^n4^ix_i\\
        &z=\sum_{i=1}^n4^{2i}x_i+\sum_{1\le i<j\le n}2\cdot 4^{i+j}y_{ij},
    \end{aligned}
    \end{equation}
    such that for any fixed $x$,
    \begin{equation*}
        \text{$\max_{y,z,\delta: (\delta,x,y,z)\in P} \delta^2-z=0$ if $x\in \{0,1\}^n$,}
    \end{equation*}
    and 
    \begin{equation*}
        \text{$\max_{y,z,\delta: (\delta,x,y,z)\in P} \delta^2-z<0$ \text{ otherwise}.}
    \end{equation*}
\end{lemma}

We can add linear constraints for an arbitrary instance of $k$-independent set problem $P'$, 
    \begin{equation}  \label{eq:indSetLin}     \tag{$P'$}
    \begin{aligned}
        & 0\le x_i\le 1,\forall i\in [n]\\
        & x_i+x_j\le 1,\forall (i,j)\in E\\
        & \sum_i x_i\ge k\,.
        \end{aligned}
    \end{equation} 
So for a small enough constant $\epsilon$, there exists a feasible point in $P\cap P'$ satisfying $\delta^2-z+\epsilon >0$ if and only if the answer to the $k$-independent set problem is yes.
The final instance $P^*$ is $P\cap P'$ with a constraint $\delta^2-z+\epsilon >0$. Note that this decision problem satisfies the form in (\ref{eq:lp2}). 

We can check that (\ref{eq:nondeg_cond}) is infeasible for $P^*$. Substituting all constant terms with 0, we get 
\begin{equation*}
    \begin{aligned}
    &x_i=0,1\le i\le n\\
    &y_{ij}\ge 0,1\le i<j\le n\\
    &z=\sum_{i=1}^n4^{2i}x_i+\sum_{1\le i<j\le n}2\cdot 4^{i+j}y_{ij}\\
    &-z>0\,,
    \end{aligned}
\end{equation*}
which is clearly infeasible.

\begin{remark}
Here we chose the independent set problem for concreteness, but in fact we can use any NP-hard problem that can be written as an integer linear program on $\{0,1\}^n$ with polynomially bounded parameters.
\end{remark}

\paragraph{Step 3: Bounding the size of reduction instance.}
Next we only need to show that $\epsilon$ can be chosen to be represented by polynomial number of bits. From Lemma~\ref{lem:d-infty} and Lemma~\ref{lem:bound-epsilon}, we can set $\epsilon$ to be any constant between $0$ and $\left(2k^{n^3+n}n!\right)^{-2}$ so that $P'$ contains a feasible integer solution if and only if $\max_{x,y,z,\delta\in P} \delta^2-z+\epsilon>0$. Indeed, if $P'$ contains an integer solution $x_0$, then fix $x_0$, we have $\max_{y,z,\delta\in P} \delta^2-z+\epsilon=\epsilon>0$ by Lemma~\ref{lem:PV91}. If $P'$ does not contain an integer solution, by Lemma~\ref{lem:d-infty} and Lemma~\ref{lem:bound-epsilon}, fixing any $x\in P'$, we have $$\max_{y,z,\delta\in P} \delta^2-z+\epsilon\le \left(2k^{n^3+n}n!\right)^{-2}+\epsilon<0.$$ The denominator for $\epsilon$ is $O(\exp(poly(n)))$ so it can be represented by polynomially many of bits.

\begin{lemma}[\cite{freund1985complexity}, Lemma 1]\label{lem:d-infty}
    Recall that \ref{eq:indSetLin} is the linear relaxation of the $k$-independent set problem for graph $G=(V,E)$ as defined above.
    If $P'\neq\emptyset$ but there is no integer solution, i.e., $P'\cap \{0,1\}^n=\emptyset$, then $d_{\infty}(P',\{0,1\}^n)\ge \left(2k^{n^3+n}n!\right)^{-1}$.
\end{lemma}

\begin{lemma}\label{lem:bound-epsilon}
    Under the linear program \ref{equ:LP-PV91} given in Lemma~\ref{lem:PV91}, for a fixed $x\in [0,1]^n$,
    \begin{equation*}
        \max_{y,z,\delta:(\delta,x,y,z)\in P} \delta^2-z\le  - d_{\infty}^2(x,\{0,1\}^n)\,.
    \end{equation*}
\end{lemma}
\begin{proof}
To reach the maximum of $\delta^2-z$, $y_{ij}$ are at their minimum, $\max\{0,x_i+x_j-1\}=(x_i+x_j-1)_+$, so
\begin{equation*}
    \max_{y} \delta^2-z = \sum_{i=1}^n 4^{2i}(x_i^2-x_i)+ \sum_{1\le i<j\le n}2\cdot 4^{i+j} \left(x_ix_j-(x_i+x_j-1)_+\right)
\end{equation*}
Note that $x_ix_j-(x_i+x_j-1)_+\le x_j(1-x_j)$, because when $x_i\le 1-x_j$, $x_ix_j-(x_i+x_j-1)_+=x_ix_j\le x_j(1-x_j)$, otherwise $x_i> 1-x_j$, $x_ix_j-(x_i+x_j-1)_+=(1-x_i)(1-x_j)\le x_j(1-x_j)$. 
Let $t_i=x_i(1-x_i)$, then 
\begin{eqnarray*}
    \max_y \delta^2-z &\le& -\sum_{i=1}^n 4^{2i}t_i + \sum_{1\le i<j\le n}2\cdot 4^{i+j} t_j\\
    &=& \sum_{i=1}^n \Big(-4^{2i} + 2\cdot 4^i\sum_{j=1}^{i-1}4^j\Big)t_i\\
    &\le & \sum_{i=1}^n -\frac{4^{2i}}{3}t_i\\
    &\le & -\min_it_i.
\end{eqnarray*}
Let $d=d_{\infty}(x,\{0,1\}^n)$.
Because $d\le 1/2$, we have
\begin{equation*}
    \min_it_i= d (1-d)\ge d^2.
\end{equation*}
\end{proof}

\begin{remark}
Note that Lemma~\ref{lem:bound-epsilon} implies Lemma~\ref{lem:PV91}, so we give an alternate shorter and more intuitive proof of the main theorem in \cite{pardalos1991quadratic}.
\end{remark}

\subsection{Hardness of Everywhere Infeasibility of PLP}\label{sec:hardness-infea}
In this section we prove hardness of everywhere infeasibility of 1-PLP. The proof is similar to feasibility and we again reduce from the independent set problem. 
\begin{theorem}
    The everywhere infeasibility of 1-PLP is NP-hard.
\end{theorem}

\begin{proof}
The everywhere infeasibility of 1-PLP can be written as follows:
\begin{equation*}
    \forall \delta\in \mathbb{R}, \forall x, A(\delta)x\ge b(\delta) \text{ is not satisfied}.
\end{equation*}
First, we take the negation, which is 
\begin{equation*}
    \exists \delta, x\text{ such that } A(\delta)x\ge b(\delta)
\end{equation*}
Note that \eqref{equ:LP-PV91} is in such form. Therefore, we can take the following program.
\begin{equation*}
    \left\{
    \begin{aligned}
        &P\\
        &\delta^2-z= 0\\
        &P'\\
    \end{aligned}
    \right.
\end{equation*}
By Lemma~\ref{lem:PV91}, the first two constraints restrict $x$ on $\{0,1\}^n$. This would solve the independent set problem combined with \eqref{eq:indSetLin}.
\end{proof}
\subsection{Everywhere Feasibility is in co-NP}
In this section we prove that everywhere feasibility of 1-PLP is in co-NP. This implies that the problem is co-NP complete.

\begin{theorem}\label{thm:evefeasibility-in-conp}
    Everywhere feasibility problem of 1-PLP is in co-NP, i.e., the following problem is in NP. Given an instance of 1-PLP $(A(\delta),b(\delta))$, decide whether exists $\delta\in \mathbb{R}$ such that there is no $x$ satisfying $A(\delta)x\ge b(\delta)$.
\end{theorem}

By linear programming duality, the above decision problem is exactly \eqref{equ:exists}, which is equivalent to the following:
\begin{equation}\label{eq:conp-condition-delta-y}
    \exists \delta,y\text{ satisfying }
    \begin{cases}
    y\ge 0\\
    A^\top (\delta)y=0\\
    b^\top (\delta)y=1
    \end{cases}
\end{equation}
Intuitively, the value of $\delta$ where $A(\delta)x\ge b(\delta)$ is infeasible should be used as the certificate. The difficulty is that $\delta$ may not even be a rational number.  
Our proof has two steps. The first step shows that if answer to the above problem is yes, we can find a $\delta_0$ that is a root of a polynomial. In the second step, we show that $y$ can be chosen to be a rational function of $\delta_0$ and the conditions in \eqref{eq:conp-condition-delta-y} can be checked within polynomial time.

\begin{lemma}\label{lem:conp-poly-delta}
    Suppose the answer to the problem in Theorem~\ref{thm:evefeasibility-in-conp} is yes. There exists $\delta_0$ satisfying the condition and $\delta_0$ is a root of a polynomial with size $poly(n,m,d,L)$.
\end{lemma}
\begin{proof}
By Corollary~\ref{cor:feasible-infeasible-on-interval}, we can divide $\mathbb{R}$ into a finite number of intervals such that $A(\delta)x\ge b(\delta)$ is always feasible or infeasible when $\delta$ is within an interval.

If $A(\delta)x\ge b(\delta)$ is infeasible except at a finite number of points on $\mathbb{R}$, then we can easily choose a rational $\delta_0$. Otherwise, there exists two consecutive intervals such that $A(\delta)x\ge b(\delta)$ is feasible on the first interval and infeasible on the second interval. Suppose the two intervals have common endpoint $\delta_1$, i.e., $A(\delta)x\ge b(\delta)$ is feasible on an open interval $(\delta_1-\epsilon,\delta_1)$ but infeasible at $\delta_1$ or feasible on $(\delta_1-\epsilon,\delta_1]$ but infeasible on an open interval $(\delta_1,\delta_2)$. 

Now we prove that $\delta_1$ is root of a polynomial with size $poly(n,m,d,L)$. Based on Lemma~\ref{lem:algo-rational-solution}, we can pick $\epsilon$ small enough so that there exists $J\subset [n]$ that $A_J^+ b_J(\delta)$ is a solution for $x$ on $(\delta_1-\epsilon,\epsilon_1)$. If $A_J^+ b_J(\delta)$ is not continuous at $\delta_1$, denominator of one of the coordinate becomes zero at $\delta_1$, which means $\delta_1$ is a root of the denominator. Now suppose $A_J^+ b_J(\delta)$ is continuous around $\delta_1$, but is not a solution in the right neighborhood of $\delta_1$. Suppose it violates a constraint $a_i^\top x\ge b_i$ in the right neighborhood of $\delta_1$. So $a_i^\top A_J^+ b_J(\delta)\ge b_i$ before $\delta_1$ but $a_i^\top A_J^+ b_J(\delta)< b_i$ after $\delta_1$. Since the function is continuous, we know $a_i^\top A_J^+ b_J(\delta_1)-b_i=0$. After multiplying any denominator on both sides, we get that $\delta_1$ is root to a polynomial with size $poly(n,m,d,L)$.

So the proof is completed if $A(\delta)x\ge b(\delta)$ is infeasible at $\delta_1$. Now assume $A(\delta)x\ge b(\delta)$ is feasible at $\delta_1$ but infeasible at $(\delta_1,\delta_2)$. Then by the above analysis, $\delta_1$ and $\delta_2$ are roots of two polynomially sized polynomial $p_1$ and $p_2$. They are both roots of $p_1p_2$. So by Lemma~\ref{lem:dis-between-roots}, $\delta_2-\delta_1\ge 2^{-poly(n,m,d,L)}$. Therefore, we can pick a rational number between $\delta_1$ and $\delta_2$ with polynomial number of bits.
\end{proof}

\begin{lemma}\label{lem:conp-poly-y}
    Suppose the answer to the problem in Theorem~\ref{thm:evefeasibility-in-conp} is yes. Let $\delta_0$ be the number chosen in Lemma~\ref{lem:conp-poly-delta}. There exists $y_0=f(\delta_0)$ satisfying \eqref{eq:conp-condition-delta-y} where $f$ is a rational function of $\delta_0$ with size $poly(n,m,d,L)$.
\end{lemma}
\begin{proof}
Note that \eqref{eq:conp-condition-delta-y} itself can be viewed as a PLP with variable $y$.
The lemma follows as a corollary of Lemma~\ref{lem:equality_solution} and  Lemma~\ref{lem:algo-rational-solution}.
\end{proof}

All the ingredients are in place to prove Theorem~\ref{thm:evefeasibility-in-conp}. 
\begin{proof}[Proof of Theorem~\ref{thm:evefeasibility-in-conp}]
The certificate of the problem are $\delta_0$ in Lemma~\ref{lem:conp-poly-delta} and $y_0$ in Lemma~\ref{lem:conp-poly-y}. By Lemma~\ref{lem:poly-check-sign}, we can represent them by irreducible polynomials and signs of their derivatives. Also, \eqref{eq:conp-condition-delta-y} can be verified in polynomial time by Lemma~\ref{lem:poly-check-sign}.
\end{proof}

\section{Application to Probabilistic Cellular Automata}\label{sec:PCA}
In this section we formally introduce the potential method developed in \cite{holroyd2019percolation} and discuss how to use the algorithm for the local feasibility of PLP to find a potential.

A one-dimensional discrete-time probabilistic cellular automaton (PCA) is a function on $\mathbb{Z}$ that evolves with time. The PCA has alphabet $A$ so at any time $t$ the configuration of the PCA is denoted by $\eta_t\in A^\mathbb{Z}$ where $\eta_t(n)$ is the value on site $n\in \mathbb{Z}$ at time $t$.

The configuration $\eta_t$ is Markov chain as $t$ changes. Given $\eta_t$, the next configuration $\eta_{t+1}$ is obtained by updating each site $n\in \mathbb{Z}$ independently. The update rule is local, random, and homogeneous. That is to say, $$\eta_{t+1}(n)=f(\eta_t(n-a),\eta_t(n-a+1),\cdots,\eta_t(n+a)),$$
where $f$ is a random function that follows a fixed distribution independent of $t$ and $n$, $a$ is a constant. Let $F$ denote the transition of the Markov chain.


The running example is the following. (A similar analysis could be applied to other PCA as well, but as discussed in Section~\ref{sec:applications}, we focus on the one elementary symmetric PCA which ergodicity is unknown.) 
The PCA has alphabet $\{0,1\}$ and noise parameter $p\in [0,1/2]$. The state is updated according to
\begin{equation*}
    \eta_{t+1}(n)=\text{BSC}_p(\NAND(\eta_t(n-1), \eta_t(n)))
\end{equation*}
where $BSC_p$ is binary symmetric channel with crossover probability $p$.
The Markov transition operator is denoted by $A_p$. The PCA can be view as a NAND function that has probability $2p$ of turning into a random bit.

The problem of interest is the ergodicity of this PCA, as defined in Definition~\ref{def:ergodicity}. The ergodicity is known to be equivalent to having positive probability of drawing in \emph{percolation game} \cite{holroyd2019percolation}. 
When $p=0$, $A_0$ is noiseless and it easy to see that $\eta_t$ may not converge. However, it is proven to be ergodic for any $p>0$ \cite{holroyd2019percolation}. 
An alternate version of noisy NAND function is NAND with edge noise, where the PCA is defined by
\begin{equation*}
    \eta_{t+1}(n)= \NAND\big(\text{BSC}_p(\eta_t(n-1)),  \text{BSC}_p(\eta_t(n))\big)
    .
\end{equation*}
The Markov transition operator is denoted by $A'_p$. 
It is shown in \cite{dobrushin1977lower} that edge noise is a more general model than vertex noise. 

We will prove the ergodicity of both PCA for small enough $p$ under the help of the algorithm for the local feasibility of PLP.
\begin{theorem}\label{thm:pca}
There exists $\epsilon>0$ such that $A_p$ is ergodic for any $p\in (0,\epsilon)$.
\end{theorem}
\begin{theorem}\label{thm:pca-edge}
There exists $\epsilon>0$ such that $A'_p$ is ergodic for any $p\in (0,\epsilon)$.
\end{theorem}

Note that the ergodicity of $A'_p$ is not known prior to this work as we know, but Theorem~\ref{thm:pca} is a corollary of the result in \cite{holroyd2019percolation}. Our contribution is to propose a new procedure of automatically finding potential functions by reducing to it to feasibility of PLP.


\subsection{Coupling and Auxiliary Chain}

A typical approach to proving the ergodicity is by coupling. Let $\eta^+_t$ and $\eta^-_t$ be two versions of the Markov chain that start at different initial configurations. If we can show that there is a coupling so that for any two initial configuration the probability of having same configuration converges to 1 as $t$ goes to infinity, then by linearity any distribution on configuration would converge to the same stationary distribution. Thus ergodicity is proven by Definition~\ref{def:ergodicity}.


The coupled Markov chain $(\eta^+_t,\eta^-_t)$ is given by the following update rule. On any site $n$, $\eta^+$ and $\eta^-$ either both apply the NAND function with probability $1-2p$, or take the same random bit with probability $2p$. The coupled chain has alphabet $\{(0,0),(1,1),(0,1),(1,0)\}$. If we ignore the difference between $(1,0)$  and $(0,1)$ and denote them by $?$, we get a new chain $F_p$ with alphabet $\{0,1,?\}$. 
The update rule of $F_p$ is illustrated in Figure~\ref{fig:chain-F-p} and the update rule of $F_p'$ is shown in Table~\ref{tab:PCA-edge}. We can define $F_p'$ for $A_p'$ in the same way. Note that the chain $F_p$ and $F_p'$ are still Markov.
\begin{figure}
    \centering
    \begin{tikzpicture}
        \node at (0,0) {1};
        \node at (1,0) {1};
        \node at (3,0) {0};
        \node at (4,0) {*};
        \node at (3,0.5) {*};
        \node at (4,0.5) {0};
        \node at (6,0) {?};
        \node at (7,0) {?};
        \node at (6,0.5) {1};
        \node at (7,0.5) {?};
        \node at (6,1) {?};
        \node at (7,1) {1};
        
        \draw (0,-0.2)--(0.5,-1);
        \draw (1,-0.2)--(0.5,-1);
        \draw (3,-0.2)--(3.5,-1);
        \draw (4,-0.2)--(3.5,-1);
        \draw (6,-0.2)--(6.5,-1);
        \draw (7,-0.2)--(6.5,-1);
        
        \node[anchor=west] at (0.3,-1.2) {1($p$)};
        \node[anchor=west] at (0.3,-1.7) {0($1-p$)};
        \node[anchor=west] at (3.3,-1.2) {1($1-p$)};
        \node[anchor=west] at (3.3,-1.7) {0($p$)};
        \node[anchor=west] at (6.3,-1.2) {1($p$)};
        \node[anchor=west] at (6.3,-1.7) {0($p$)};
        \node[anchor=west] at (6.3,-2.2) {?($1-2p$)};
    \end{tikzpicture}
    \caption{Update rule of $F_p$. Here * means arbitrary symbol from $\{0,1,?\}$.  }
    \label{fig:chain-F-p}
\end{figure}
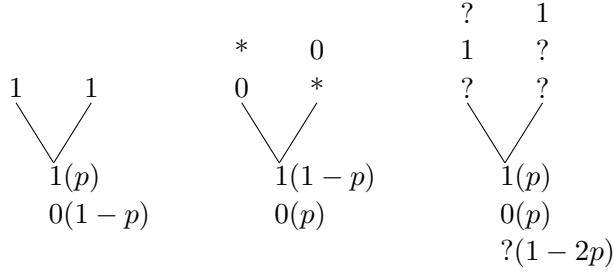

By definition, restricting $F_p$ to not contain $?$ is exactly $A_p$. So if $F_p$ is ergodic then $A_p$ is ergodic. In fact, we can further reduce the ergodicity of $F_p$ \cite{holroyd2019percolation}.
\begin{lemma}[\cite{holroyd2019percolation}]\label{lem:pca-aux-chain}
    $F_p$ is ergodic if
    any stationary distribution of $F_p$ has 0 probability of symbol $?$.
\end{lemma}
The same statement can be proven for $F_p'$, which can be found in Appendix~\ref{sec:pca-aux-chain-edge} for completeness.
\begin{lemma}\label{lem:pca-aux-chain-edge}
    $F_p'$ is ergodic if
    any stationary distribution of $F_p$ has 0 probability of symbol $?$.
\end{lemma}

\subsection{Potential Function Method}
We say the distribution of a configuration $\eta(n)$ is \emph{shift-invariant} if $\cdots,\eta(-1),\eta(0),\eta(1),\cdots$ and $\cdots,\eta(-1+k),\eta(0+k),\eta(1+k),\cdots$ have the same distribution for any $k\in \mathbb{Z}$. It is without loss of generality to only consider stationary distributions that are shift-invariant.
For $s\in \{0,1,?\}^n$, and a shift-invariant distribution $\mu$ on configurations $\eta$, let $\mu(s)=\mu(\eta:(\eta(0),\eta(1),\cdots,\eta(n))=s)$. A potential function is a weighted sum of the probability of strings, defined formally as follows.
\begin{definition}[Potential function]\label{def:potential-function}
    A \emph{potential function} $w$ with length $\ell$ is a vector in $\mathbb{R}^{3^l}$, where the coordinates are labeled by $\{0,1,?\}^\ell$. We write $\{s\}$ for the $s$th standard basis vector, and
    \begin{equation*}
        w=\sum_{s\in S}c_s \{s\}\,.
    \end{equation*}
    where $S\subset \{0,1,?\}^\ell$.
    For a shift-invariant distribution $\mu$, define
    \begin{equation*}
        \{s\}(\mu) = \mu(s)\,,
    \end{equation*}
    and
    \begin{equation*}
        w(\mu) = \sum_{s\in S}c_s\mu(s)\,.
    \end{equation*}
\end{definition}
The potential function method can be summarised as follows:
\begin{lemma}[\cite{holroyd2019percolation}]\label{lem:pca-potential-function-method}
To show that stationary distribution of $F_p$ has 0 probability of symbol $?$, it suffices to design a \emph{potential function} $w$
such that for any shift-invariant distribution $\mu$,
\begin{equation}\label{eq:pca-weight-func-condition}
    w(F_p\mu) \le w(\mu) - \mu(?s_0)
\end{equation}
where $s_0$ is an arbitrary string with alphabet $\{0,1\}$.
\end{lemma}
\begin{proof}
Because we can assume without loss of generality that the stationary distribution $\mu$ is shift-invariant, $w(\mu)=w(F_p\mu) \le w(\mu) - \mu(?s_0)$ means that $\mu(?s_0)=0$. Since $\mu(?s_0)=F_p\mu(?s_0)\ge p^{|s_0|}F_p\mu(?)=p^{|s_0|}\mu(?)$, we can conclude that $\mu(?)=0$.
\end{proof}


\subsection{Connection with PLP}
In \cite{holroyd2019percolation}, a carefully designed potential function was introduced. In contrast, we will show how to find a potential function that satisfies \eqref{eq:pca-weight-func-condition} with local feasibility of PLP. This connection between finding a potential function and PLP was introduced in \cite{makur2020broadcasting}, but \cite{makur2020broadcasting} only analyzed the PLP for specific choices of parameter $p$, which are normal LP problems.

First let us focus on the relation between $w(\mu)$ and $w(F_p\mu)$. Suppose $w=\sum_{s\in S}c_s\{s\}$ and without loss of generality any $s\in S$ has length $\ell$. For simplicity we say the \emph{length} of $w$ is $\ell$. As stated in Definition~\ref{def:potential-function}, $w$ can be viewed as a vector with coordinates labeled by $\{0,1,?\}^\ell$, also denoted by $w$. We will show $w(F_p\mu)=w'(\mu)$ where $w'$ is a potential function with length $\ell+1$. This can be described by the following lemma.
\begin{lemma}\label{lem:pca-transition-of-potential}
    For any potential function $w$ with length $\ell$, and any PCA, $F_p$, with alphabet $\mathcal{A}$ and the following form:
    \begin{equation*}
        \eta_{t+1}(n) = f (\eta_t(n-a),\eta_t(n-a+1),\cdots,\eta_t(n+b)),
    \end{equation*}
    there exists a potential function $w'$ such that $w(F_p\mu)=w'(\mu)$. Further, define matrix $C_\ell\in \mathbb{R}^{|\mathcal{A}|^{\ell+a+b}}\times \mathbb{R}^{|\mathcal{A}|^{\ell}}$ that for $s'\in \mathcal{A}^{\ell+a+b}$ and $s\in \mathcal{A}^{\ell}$,
    \begin{align*}
        C_\ell(s',s) &= \Pr[(\eta_{t+1}(1),\cdots, \eta_{t+1}(\ell))=s|(\eta_{t}(1-a),\cdots, \eta_{t}(\ell+b))=s']\\
        &= \prod_{n=1}^{\ell} \Pr[f(s'(n-a),s'(n-a+1),\cdots,s'(n+b))=s(n)],
    \end{align*}
    we have
    \begin{equation*}
        w' = C_\ell w.
    \end{equation*}
\end{lemma}
\begin{proof}
\begin{eqnarray*}
w(F_p\mu) &=& \sum_{s\in S }c_s (F_p\mu)(s)\\
&=&\sum_{s\in S }c_s\sum_{s\in \{0,1,?\}^{\ell+1}}\mu(s')\Pr[(\eta_{t+1}(1),\cdots, \eta_{t+1}(\ell))=s|(\eta_{t}(1-a),\cdots, \eta_{t}(\ell+b))=s'] \\
&=& \sum_{s\in S }\sum_{s\in \{0,1,?\}^{\ell+1}}C_\ell(s,s)c_s\mu(s') = C_\ell w(\mu)
\end{eqnarray*}
\end{proof}

Note that in the case of $F_p$, we have $a=1,b=0$ and $\mathcal{A}=\{0,1,?\}$. So the size of $C_l$ is $3^{\ell+1}\times 3^\ell$. Also, $C_\ell$ is a polynomial of $p$, see Appendix~\ref{app:matrix} for the exact expression of $C_\ell$ for $F_p$. 

For a periodic configuration $\eta$ that repeats string $y$ and a potential function $w=\sum_{s\in S}c_s\{s\}$, we use $w(y)$ to denote the value of the potential.
\begin{definition}\label{def:pca-weight-function-on-string}
    For a potential function $w=\sum_{s\in S}c_s\{s\}$ with length $l$ and a string $y\in \cup_{n\ge \ell}\{0,1,?\}^n$, define
    \begin{equation*}
        w(y)=\sum_{s\in S}c_s\sum_{i=0}^{|y|}\mathbf{I}\{(s_{i\mod |y|},s_{i+1\mod |y|}\cdots,s_{i+|s|\mod |y|})=s\}
    \end{equation*}
\end{definition}
\begin{lemma}\label{}
    A sufficient condition for \eqref{eq:pca-weight-func-condition} is that for any string $y$ with alphabet $\{0,1,?\}$,
\begin{equation}\label{eq:pca-vector-condition}
    C_\ell w(y) \le w(y) - \{?s_0\}(y).
\end{equation}
\end{lemma}
\begin{proof}
First, let us show that for any shift-invariant measure $\mu$, the probability of any aperiodic configuration is 0. Suppose for contrary that this is not true. Because there are countably many aperiodic configurations, there must exists one with positive probability. Suppose $\mu(\eta(n))>0$ where $\eta$ is an aperiodic configuration. Since $\mu$ is shift-invariant, $\mu(\eta(n))=\mu(\eta(n+k))$ for any $k\in \mathbb{Z}$. This leads to a contradiction because $\mu(\{\eta(n+k):k\in \mathbb{Z}\})$ would be $+\infty$.

Therefore, it suffices to show that \eqref{eq:pca-vector-condition} leads to \eqref{eq:pca-weight-func-condition} on $\mu$ that has probability 1 on periodic configurations.  Suppose $L$ is the period of a periodic configuration $\eta$, so $\eta(n+L)=\eta(n)$, let $\mu_{\eta} = 1/L\sum_{i=1}^L\delta_{\eta_{n+i}}$. Here $\delta_{\eta_{n+i}}$ stands for the delta distribution on $\eta_{n+i}$. Any shift-invariant distribution would be a convex combination of $\mu_{\eta}$ for different $\eta$. Let $y=(\eta(1),\eta(2),\cdots,\eta(L))$. For any potential function $w=\sum_{s\in S}c_s\mu(s)$, we can conclude by Definition~\ref{def:pca-weight-function-on-string},
$$w(\mu)=1/L\sum_{i=1}^L\mathbb{I}[(\eta(i),\eta(1+i),\cdots,\eta(|s|+i))=s]=w(y).$$
Also, by Lemma~\ref{lem:pca-transition-of-potential} we have 
$$w(F_p\mu ) = C_\ell w(\mu)= C_\ell w(y).$$
So \eqref{eq:pca-weight-func-condition} is equivalent to 
$$C_\ell w(y)\le w(y)-\{?s_0\}(y).$$
\end{proof}

Note that $\{s\}=\{s0\}+\{s1\}+\{s?\}$. We can extend a potential function to any longer length. Use $P_\ell (w)$ to denote the length-$\ell$ version of $w$ if $w$ has length not greater than $\ell$. $P_\ell$ is a linear transformation. Specifically, when $w$ is expressed as a vector with length $\ell_0$, $P_\ell$ is matrix $P_\ell=(1,1,1)^{\top\otimes\ell-\ell_0}\otimes  I_{\ell_0}$, where $\otimes$ stands for  Kronecker product. Therefore, if $|s_0|\le \ell$ \eqref{eq:pca-vector-condition} is equivalent to that for any string $y$,
\begin{equation}\label{eq:pca-final-condition}
    \left(P_{\ell+1}(w)-C_\ell w -P_{\ell+1}(\{?s_0\})\right)(y)\ge 0.
\end{equation}

This leads to the problem of what potential function $w$ is non-negative over any string. Our goal is to introduce a sufficient condition for \eqref{eq:pca-vector-condition} expressed by a PLP . A trivial sufficient condition would be that for $w=\sum_{s\in S}c_s\{s\}$, every coefficient $c_s\ge 0$.  However, the condition is too strong that no potential would satisfy. To build intuition, let us construct a directed weighted graph $G$ with length $l$ strings as vertices and there is an edge between two vertices $s_0$ and $s_1$ if $(s_0(1),s_0(2),\cdots s_0(l-1))=(s_1(0),s_1(1),\cdots s_1(l-2))$. Each vertex $s$ has weight $c_s$ if $s\in S$ and 0 otherwise. Then by Definition~\ref{def:pca-weight-function-on-string}, $w(y)$ corresponds to the weight of a cycle on the graph. So $w(y)\ge 0$ for any $y$ is equivalent to having no negative cycles on $G$. 

\cite{makur2020broadcasting} proved a sufficient condition for this expressed by linear constraints.
\begin{lemma}[\cite{makur2020broadcasting}, Proposition 10]\label{lem:key-makur20}
    There exists $A_{\ell},B_\ell$ of size $3^{\ell+1}\times 3^{\ell}$ that only depends on $\ell$ such that
    for a potential function $w$ with length $\ell+1$, if there exists $z$ such that
    \begin{equation*}
        A_\ell z\le B_\ell w,
    \end{equation*}
    we have that for any string $y$,
    \begin{equation*}
        w(y)\ge 0.
    \end{equation*}
\end{lemma}
The concrete expression of $A_l$ and $B_l$ is included in Appendix~\ref{app:matrix}. 

\begin{proof}[Proof of Theorem~\ref{thm:pca} and Theorem~\ref{thm:pca-edge}]
From Lemma~\ref{lem:pca-transition-of-potential}, the coefficients in \eqref{eq:pca-final-condition} are polynomials of the noise parameter $p$. So by Lemma~\ref{lem:key-makur20}, a sufficient condition for  Theorem~\ref{thm:pca} is to check the local feasibility of the following PLP for a certain length $\ell$ and a certain $s_0$ where $|s_0|\le  \ell$:
\begin{equation}\label{eq:PCA-final-PLP}
    A_{\ell+1} z \le B_{\ell+1} \left(P_{\ell+1}(w)-C_\ell w -P_{\ell+1}(\{?s_0\})\right).
\end{equation}
Here $z$ and $w$ are variable vectors. 

\paragraph{Finding a potential.} The rest of proof is carried out using a computer. To use local feasibility of PLP to find a potential funciton satisfying \eqref{eq:PCA-final-PLP}, we can search in the space of possible length $\ell$ and choice of $s_0$, using the algorithm in Section~\ref{thm:1-dim-local}. The local feasibility of \eqref{eq:PCA-final-PLP} for any choice of $\ell$ and $s_0$ would lead to ergodicity. Note in Lemma~\ref{lem:pca-transition-of-potential}, the size of $C_\ell$ grows exponentially with length of the potential function. So only small $\ell$ yields a PLP that is computationally efficient in practice. We tried $\ell=2$ or $3$ and all possible $s_0$ with length no more than $\ell$. In the end, we find that \eqref{eq:PCA-final-PLP} is locally feasible for $F_p$ with $\ell=3$ and $s_0=01$.  It is locally feasible for $F_p'$ with $\ell=3$ and $s_0=10$. This finishes the proof for Theorem~\ref{thm:pca} and Theorem~\ref{thm:pca-edge}. See the potential functions in Appendix~\ref{sec:potential-pca}.

\end{proof}


\subsection{Attempt on Soldier's Rule}\label{sec:soldiers-rule}
Let us first recall the definition of Soldier's Rule. In Soldier's Rule, the alphabet is $\{-1,1\}$, indicating the direction a soldier is facing. At each time step, each soldier sets their new direction based the majority vote of the direction of itself and the direction of the first neighbor and the third neighbor in their direction. The result then goes through a BSC channel. Thus,
$$\eta_{t+1}(n)= \text{BSC}_p\left(\text{Maj} (\eta_t(n),\eta_t(n+\eta_t(n),\eta_t(n+3\eta_t(n))\right).$$
The same argument can be applied to Soldier's Rule to show that a sufficient condition of ergodicity is \eqref{eq:PCA-final-PLP}.

However, by Lemma~\ref{lem:pca-transition-of-potential}, here the size of $C_\ell$ is $3^{r+6}\times 3^{r}$. Therefore, even for potential with length 3, the final PLP would have $3^{10}$ number of constraints, which makes the algorithm inefficient in practice. We leave it as an open question whether the ergodicity of soldier's rule under BSC noise can be solved efficiently with our approach. This is potentially possible either by improving the running time for Theorem~\ref{thm:1-dim-local} in general, or exploring special structure of \eqref{eq:PCA-final-PLP} so that the algorithm in Theorem~\ref{thm:1-dim-local} can be accelerated.

\section{Application in Broadcasting of Information on 2D Grid}\label{sec:broadcasting-on-grid}

In information theory, we are interested in how information broadcasts on a infinite DAG with noise. One important case is when the DAG is a 2D grid. To keep a consistent notation with section~\ref{sec:PCA}, the nodes on grid are indexed by $(t,i)$ where $t\ge0$, $0\le i\le t$. The bit at coordinate $(t,i)$ is denoted by $\eta_t(i)$. At time 0, there is one bit of information at coordinate $0,0$, either 0 or 1. At any time $t>0$, node at $(t-1,i)$ receives information sent by previous layer and pass a bit to $(t,i)$ and $(t,i+1)$. The edge between two nodes is noisy, specifically, a BSC channel with noise parameter $p$. 


For nodes on the boundary, i.e., nodes with coordinates $(t,t)$ and $(t,0)$, they only receive one bit of input. So we assume that they pass on whatever bit they receive. For other nodes, they receive two inputs and output to two other nodes. The output is a boolean function of the two-bit input, and different functions $f$ corresponds to different Markov chains.
\begin{equation}
    \eta_t(i) = f\left( BSC_p(\eta_{t-1}(i-1)), BSC_p(\eta_{t-1}(i)) \right)
\end{equation}
We will use $\eta_t^+$ and $\eta_t^-$ to denote the Markov chain that starts with bit 1 and bit 0 respectively. Our goal is to prove that for any function $f$, $\eta_t^+$ and $\eta_t^-$ converges to the same distribution, i.e., 
\begin{equation*}
    \lim_{t\rightarrow \infty} ||P_{\eta_t^+},P_{\eta_t^-}||_{TV}=0.
\end{equation*}
The information of the initial bit is lost at infinity if this is true, so we say the recovery is not possible. 

This problem can also be approached by potential function method similar to PCA and then be reduced to the local feasibility of a PLP. Results in this section are shown in \cite{makur2020broadcasting}. We include the results here to show its connection with our work.

The nontrivial choices for $f$ are AND, XOR, NAND and IMP where IMP is the following function. 
\begin{center}
    \begin{tabular}{|c|c|c|}
    \hline
    $x_1$    & $x_2$ & IMP$(x_1,x_2)$ \\
    \hline
     0    & 0& 1\\
     0&1&1\\
     1&0&0\\
     1&1&1\\
     \hline
    \end{tabular}
\end{center}
Other non-constant functions are equivalent to these four functions, or only depends on one of the inputs.

It is conjectured that it is impossible to recover for any function.
In \cite{makur2020broadcasting}, the impossibility of recovery was proven for AND and XOR for any $p$. They also proved it for NAND at many specific $p$, which corresponds to PLP at a fixed parameter. Here we will complete the conjecture for small error probability for NAND and IMP and prove the following theorem.

\begin{theorem}\label{thm:broadcasting-on-grid}
    There exists $\epsilon>0$ such that for any $p\in (0,\epsilon)$,
    the recovery is not possible for any binary function $f$.
\end{theorem}

Intuitively, if a Markov chain converges for small error probability, it should also converge for larger error probability. However, we are not able to prove the monotonicity of convergence with respect to $p$ formally.

\subsection{Auxiliary Chain and Potential Function}
Similar to PCA, we can define a coupled chain between $\eta_t^+$ and $\eta_t^-$ so that at any edge, the two chain becomes the same random bit with probability $2p$ and remains the initial bit with probability $1-2p$. The coupled chain can be simplified with alphabet $\{0,1,?\}$ where $?$ means that the two chain are not equal at that site, see Table~\ref{tab:PCA-edge}. We always denote the configuration by $\eta_t$ for simplicity because the same analysis holds for NAND and IMP.


Therefore, it suffices prove that as $t$ tends to infinity, $\eta_t$ consists of 0 and 1 with probability converging to 1. In other words,
\begin{equation}\label{eq:broadcasting-suf-condition}
    \lim_{t\rightarrow \infty}\E[\text{number of ? in }\eta_t] = 0
    .
\end{equation}

To this end, we introduce the potential function in the setting of broadcasting on grid.
\begin{definition}[potential function]\label{def:broadcasting-potential}
A potential function $w$ of length $\ell$ is written as a vector
\begin{equation*}
    w=\sum_{s\in \{0,1,?\}^\ell}c_s \{s\},
\end{equation*}
its value on a string $\eta$ with length at least $\ell$ is defined as
    \begin{equation*}
        w(\eta)=\sum_{s\in \{0,1,?\}^\ell}c_s\sum_{i=1}^{|\eta|-\ell+1}\mathbf{I}\{(\eta_{i},\eta_{i+1}\cdots,\eta_{i+|s|-1})=s\}
    \end{equation*} For a distribution $\mu$ on string $\eta$, 
\begin{equation*}
    w(\mu)=\E_{\eta\sim \mu}[w(\eta)].
\end{equation*}
\end{definition}

Therefore, the statement \eqref{eq:broadcasting-suf-condition} is equivalent to 
\begin{equation*}
    \lim_{t\rightarrow \infty}\{?\}(\eta_t) = 0.
\end{equation*}


The following lemma describes the potential function we need for the condition above to hold.
\begin{lemma}[Proposition 4, \cite{makur2020broadcasting}]\label{lem:broadcasting-condition-on-potential}
    A sufficient condition for $\lim_{t\rightarrow \infty}\{?\}(\eta_t) = 0$ is to find a function $w$ such that
\begin{enumerate}
    \item For any $\eta_t$, $w(\eta_{t+1})\le w(\eta_t)$
    \item For any string $y$, $w(y)\ge \{?\}(y)$.
\end{enumerate}
\end{lemma}

\subsection{Connection with PLP}
The goal is to again reduce sufficient condition in Lemma~\ref{lem:broadcasting-condition-on-potential} to local feasibility of a PLP. 

Intuitively we want to use similar conclusion as Lemma~\ref{lem:pca-transition-of-potential} 
and Lemma~\ref{lem:key-makur20}. The main difference between PCA and broadcasting on grid is that grid has a boundary while PCA is infinite. To circumvent the problem, let us introduce the following concepts.
\begin{definition}
A potential function $w=\sum_{s\in S}c_s\{s\}\ (c_s\not=0)$ is said to be $?$-only if any $s\in S$ contains at least one $?$ symbol.
\end{definition}

\begin{definition}
An $l$-boundary-coupled string is a string with alphabet $0,1,?$ such that the first $l$ and last $l$ symbols do not contain $?$.
\end{definition}

It is observed in \cite{makur2020broadcasting} that we can only consider $l$-boundary-coupled strings as for any $l$, $\eta_t$ will become $l$-boundary-coupled with probability converging to 1 when $t$ goes to infinity.

The following lemma is proven in \cite{makur2020broadcasting}. The proof is similar to Lemma~\ref{lem:pca-transition-of-potential}.
\begin{lembis}{lem:pca-transition-of-potential}
For a $?$-only potential function $w$ with length $l$, there exists a potential function $w'$ such that $w(\eta_{t+1}) = w'(\eta_t)$.
Further, define matrix $C_\ell\in \mathbb{R}^{3^{\ell+1}}\times \mathbb{R}^{3^{\ell}}$ that for any $l$-boundary-coupled string $s\in \{0,1,?\}^{\ell+1}$ and $s'\in \{0,1,?\}^{\ell}$,
    \begin{equation*}
        C_\ell(s,s') = \Pr[(\eta_{t+1}(1),\cdots, \eta_{t+1}(\ell))=s'|(\eta_{t}(0),\cdots, \eta_{t}(\ell+1))=s],
    \end{equation*}
    we have
    \begin{equation*}
        w' = C_\ell w.
    \end{equation*}
\end{lembis}

Observe that for $?$-only potential functions, Definition~\ref{def:broadcasting-potential} and Definition~\ref{def:pca-weight-function-on-string} are equivalent on any $l$-boundary-coupled strings.
Therefore, we have the following corollary of Lemma~\ref{lem:key-makur20}:
\begin{lembis}{lem:key-makur20}
    There exists $A_{\ell},B_\ell$ of size $3^{\ell+1}\times 3^{\ell}$ that only depends on $\ell$ such that
    for a $?$-only potential function $w$ with length $\ell+1$, if there exists $z$ such that
    \begin{equation*}
        A_\ell z\le B_\ell w,
    \end{equation*}
    we have that for any $l$-boundary-coupled string string $y$,
    \begin{equation*}
        w(y)\ge 0.
    \end{equation*}
\end{lembis}

\begin{proof}[Proof of Theorem~\ref{thm:broadcasting-on-grid}]
Using Lemma~\ref*{lem:pca-transition-of-potential}'
we can rewrite the conditions in Lemma~\ref{lem:broadcasting-condition-on-potential} for a ?-only potential function.
\begin{equation}\label{eq:broadcasting-condition}
    \begin{aligned}
        &\text{For any $l$-boundary-coupled string $y$, we have}\\
        &(w-C_\ell w) (y)\ge 0, (w-\{?\})(y)\ge 0\\
    \end{aligned}
\end{equation}

Then we can use Lemma~\ref*{lem:key-makur20}' to get a sufficient condition of \eqref{eq:broadcasting-condition}. That is, \eqref{eq:broadcasting-condition} holds if the following linear program is feasible.
\begin{equation}
    \begin{aligned}
        A_{l+1}z \le B_{l+1} (P_{\ell+1}(w) - C_\ell w) \\
        A_{l+1} z \le B_{l+1} (P_{\ell+1}(w) - P_{\ell+1} (\{?\})) \\
        w \text{ is ?-only}\\
    \end{aligned}
\end{equation}


The explicit expression for $A_\ell$, $B_\ell$ and $C_\ell$ can be found in Appendix~\ref{app:matrix}.
Because the linear program is parameterized by polynomials of $p$, this can be solved by our algorithm. We find a potential function with length 3 for both IMP and NAND.

For NAND function, the potential function found by the algorithm is the following:
\begin{align}\label{eq:pot_nand}
\begin{split}
    w&=-\frac{\delta}{1+509\delta/216}\{00?\}+\frac{2+3\delta}{1+509\delta/216}\{01?\}+\frac{4\delta}{1+509\delta/216}\{0?0\}+\frac{2+2\delta}{1+509\delta/216}\{0?1\}\\
    &+\frac{2+2\delta}{1+509\delta/216}\{0??\}+\frac{\delta}{1+509\delta/216}\{10?\}+\frac{1+\delta}{1+509\delta/216}\{11?\}+\frac{1-\delta}{1+509\delta/216}\{1?0\}\\
    &+\frac{2}{1+509\delta/216}\{1?1\}+\frac{2+3\delta/2}{1+509\delta/216}\{1??\}+\frac{1+223\delta/108}{1+509\delta/216}\{?00\}+\frac{1+439\delta/108}{1+509\delta/216}\{?01\}\\
    &+\frac{1+2\delta}{1+509\delta/216}\{?0?\}+\frac{2-4\delta}{1+509\delta/216}\{?10\}+\frac{1+73\delta/432}{1+509\delta/216}\{?11\}+\frac{2-\delta}{1+509\delta/216}\{?1?\}\\
    &+\frac{1-4\delta}{1+509\delta/216}\{??0\}+\frac{2-29\delta/48}{1+509\delta/216}\{??1\}+\frac{2-\delta}{1+509\delta/216}\{???\}
\end{split}
\end{align}

For IMP, the potential function is 
\begin{align}\label{eq:pot_imp}
\begin{split}
    w&=(40+6175\delta/38)\{00?\}+(80+611\delta/19)\{01?\}+40\{0?0\}+80\{0?1\}\\
    &+(80+200\delta/19)\{0??\}+(80-1611\delta/19) \{10?\}+ (40-\delta) \{11?\}+80\{1?0\}\\ 
    &+(80+7521\delta/38)\{1?1\}+(120+11807\delta/38)\{1??\}+40\{?00\}+80\{?01\}+80\{?0?\}\\
    &+(-40+1481\delta/38)\{?10\}+(-40-2941\delta/38)\{?11\}+(-20-747\delta/19)\{?1?\}\\ &-19\{??0\}+(10+2373\delta/19)\{??1\}+(40+5138\delta/19)\{???\}
\end{split}
\end{align}

\end{proof}

\section{Efficient Solution with Only Equality Constraints}\label{sec:equality}
In this section we will prove Theorem~\ref{thm:equality}.
Although we showed that the problem of local feasibility of 2-PLP is not tractable, it turns out that the equality version is solvable. For a linear system
\begin{equation}\label{eq:equality}
    A(\delta) x = b(\delta),
\end{equation}
where $A$ and $b$ are polynomials parameterized by $\delta\in \mathbb{R}^d$,
we are interested in whether there exists $x\in \mathbb{R}^n$ satisfying the equalities when $\delta$ is around the origin.

\begin{theorem}
    The following decision problem can be solved in polynomial time assuming $d$ is constant. Input an $m\times n$ matrix $A$ and an $m$-dimensional vector $b$ as polynomials of $\delta\in \mathbb{R}^d$. Is there an $\epsilon>0$ such that for all $ \delta\in B_\epsilon(0)$, there exists $x$ satisfying $A(\delta) x = b(\delta)$?
\end{theorem}

Note that this is different than the problem of solving a linear system in the quotient field of $\mathbb{Q}[\delta]$. A general solution $x(\delta)$ in the field that satisfies \eqref{eq:equality} may not be defined on some value of $\delta$. The most simple example would be $\delta x = 1$ where $\delta$ is 1-dimensional. A solution exists in the field $x=1/\delta$ but the equation is still infeasible when $\delta=0$. Also, a general solution not being defined at $\delta_0$ does not mean that the system is infeasible at $\delta_0$, with $\delta^2x=\delta$ as an example. 

The idea is to divide the space of $\delta$ into polynomially many regions such that at each region $x$ can be expressed as a rational function of $\delta$. The key is to find regions where the psudoinverse 
of $A$ is continuous. 

\begin{definition}
    A \emph{real algebraic set} in $\mathbb{R}^d$ is a set of form
    \begin{equation*}
        V(F)=\{\delta\in \mathbb{R}^d:p(\delta)=0,\ \forall p\in F\}
    \end{equation*}
    where $F$ is a set of polynomials. This is called the variety of $F$. Use $V(p)$ to denote the variety of a single polynomial, $\{\delta\in \mathbb{R}^d:p(\delta)=0\}$.
\end{definition}

\subsection{Preliminaries in Algebraic Geometry}

\begin{definition}
    An \emph{irreducible algebraic set} $S$ is an algebraic set that cannot be the union of two algebraic sets $S_1$ and $S_2$ such that both of them is neither $\emptyset$ nor $S$.
    
    An \emph{irreducible component} of an algebraic set $S$ is a subset of $S$ that is an irreducible algebraic set.
\end{definition}



\begin{lemma}\label{lem:equality-intersection-of-varieties}
    For any $d$ irreducible polynomials $p_1,p_2,\cdots,p_d$ in ring $\mathbb{Q}[\delta_1,\delta_2,\cdots,\delta_d]$, any irreducible component of $\cap_{i=1}^dV(p_i)=V(\{p_i:i\in [d]\})$ is $0$-dimensional, i.e., a single point on $\mathbb{R}^d$. 
\end{lemma}
\begin{proof}
First choose any irreducible branch of $V(\{p_i:i\in [d]\})$ and suppose its ideal is $I$.
From Theorem 1.8A in \cite{hartshorne2013algebraic}, $\text{dim}(\mathbb{Q}[\delta_1,\delta_2,\cdots,\delta_d]/I) = d - \text{ht}(I) $ where $\text{ht}$ stands for height of an ideal. We have $(0)\subsetneq (p_1)\subsetneq (p_1,p_2)\subsetneq\cdots(p_1,\cdots,p_d)\subseteq I$. Also, $\text{ht}(I)\le d$. So the height of $I$ is $d$, which means $$\text{dim}(\mathbb{Q}[\delta_1,\delta_2,\cdots,\delta_d]/I)=0.$$
The extension from $\mathbb{Q}[\delta_1,\delta_2,\cdots,\delta_d]$ to $\mathbb{C}[\delta_1,\delta_2,\cdots,\delta_d]$ preserves dimension, so $\mathbb{C}[\delta]/I$ is 0-dimensional. This means $\mathbb{R}[\delta]/I$ can only be 0-dimensional
\end{proof}

\subsection{Algorithm}
Recall that $Ax=b$ is feasible if and only if $AA^+b-b=0$ (Lemma~\ref{lem:pseudo_inverse}). So a direct approach is to calculate $A^+$.
The main issue is that $A^+_J(\delta_1,\delta_2,\cdots ,\delta_d)$ is not necessarily a continuous function on $\mathbb{R}^d$. But we will show that this is true if we further divide $\mathbb{R}^d$ into polynomially many regions where $A^+_J$ is a rational function on each of them. 

\begin{lemma}\label{lem:equality-division-of-space}
    There exists an $\epsilon>0$ such that $B_\epsilon(0)$ can be divided into polynomially many regions and on each region $A^+(\delta)$ is a continuous rational function. $A^+$ on each region can be computed in polynomial time (given as a rational function).
\end{lemma}

\begin{proof}

To this end, we use the following expression for $A^+_J$ again.
\begin{equation*}
\begin{aligned}
    A^+_J &= \lim_{y\rightarrow 0}(A^TA+y I)^{-1}A^T_J\\
    &=\lim_{y\rightarrow 0} \frac{1}{\mathrm{det}(A^TA+y I)}\mathrm{adj}(A^TA+y I)A^T_J\,,
\end{aligned}
\end{equation*}
where $\mathrm{adj}(B)$ denotes the adjugate of a matrix $B$, equal to the transpose of the cofactor matrix. 

Note that $\text{det}(A^TA+y I)$ and  $\text{adj}(A^TA+y I)A^T_J$ are polynomials in $\delta_1,\delta_2$ and $y$. Now assume 
\begin{equation*}
\begin{aligned}
    &\text{det}(A^TA+y I) = q_0(\delta)+yq_1(\delta)+y^2q_2(\delta)+\cdots y^kq_k(\delta)\\
    &\text{adj}(A^TA+y I)A^T_J = p_0(\delta)+yp_1(\delta)+y^2p_2(\delta)+\cdots y^kp_k(\delta)\,,
\end{aligned}
\end{equation*}
where $p_i$ are matrices and $k\le m$. An alternative expression for $A^+_J$ is the following:
\begin{equation}\label{eq:equality-pseudo-inverse-with-region}
    A^+_J(\delta) = 
    \begin{cases}
        \frac{p_0(\delta)}{q_0(\delta)},&\text{when }q_0(\delta)\not=0\\
        \frac{p_1(\delta)}{q_1(\delta)},&\text{when }q_0(\delta)=0,q_1(\delta)\not=0\\
        \cdots\\
        \frac{p_k(\delta)}{q_k(\delta)},&\text{when }q_0(\delta)=q_1(\delta)=\cdots=q_{k-1}(\delta)=0\\
    \end{cases}
\end{equation}
An important observation is that Suppose on a certain region $R\subset \mathbb{R}^d$, for any $1\le i\le k$, whether $q_i=0$ is fixed. Then $A^+(\delta)$ can be expressed as a rational function on $R$, given by the expression above. 

Now take all irreducible factors of non-zero polynomials from $q_0,q_1,\cdots,q_k$, suppose they are $f_1,\cdots,f_r$. The degree of $q_i$ is bounded by $md$, so $r\le kmd\le m^2d$. There is a line of research for polynomial-time algorithms for factorizing in $\mathbb{Q}[\delta_1,\delta_2,\cdots \delta_d]$, a survey can be found in \cite{kaltofen1992polynomial}. So the running time of this step is bounded by a polynomial.

Note that if $f_i(0)\not=0$, $f_i$ would be non-zero on $B_\epsilon(0)$ for a small enough $\epsilon$ since it is continuous. So we only care about varieties of all $f_i$ where $f_i(0)=0$. Suppose the varieties of these $f_i$'s are $V_1,\cdots ,V_{r'}$. By Lemma~\ref{lem:equality-intersection-of-varieties}, the intersection of any $d$ of the $V_i$'s is either empty or finite number of points on $\mathbb{R}^2$. 
Therefore, we can choose $\epsilon$ small enough such that the intersection of any $d$-tuple of $V_i$ on $B_\epsilon(0)$ is either empty or equal to $\{0\}$. Let $V_S$ where $S\subset [r]$ denote $\cap_{i\in S}V_i$. Now we can partition $B_\epsilon(0)$ into the following regions:
\begin{equation}\label{eq:equality-regions}
\begin{aligned}
    &\{0\},B_\epsilon(0)\backslash (\cup_{|S|\le d,S\subset [r]}V_S),\\
    &B_\epsilon(0)\cap V_S,\ \forall |S|\le d-1,S\subset [r]
    .
\end{aligned}
\end{equation}
$f_1,\cdots,f_r$ are either always zero or always non-zero on any region, so the same holds for $q_1,\cdots,q_k$. So on each region $A^+$ is given by \eqref{eq:equality-pseudo-inverse-with-region}.
\end{proof}

Now we are ready to prove Theorem~\ref{thm:equality}.
With Lemma~\ref{lem:equality-division-of-space}, $B_\epsilon(0)$ is partitioned into polynomially many regions in \eqref{eq:equality-regions}. On each region $A^+$ is a continuous rational function and so is $AA^+b-b$.
We can check the feasibility of $Ax=b$ by checking whether the following rational function $AA^+b-b$ is zero on every region. There are three cases:
\begin{enumerate}
    \item The region is $\{0\}$, $A$ and $b$ are fixed. So the system becomes a linear system without parameter.
    \item The region is $B_\epsilon(0)\backslash (\cup_{|S|\le d,S\subset [r']}V_S)$, assume $AA^+b-b=g(\delta)$. The region is $d$-dimensional, so $g(\delta)=0$ on this region if and only if $g=0$.
    \item The region is $B_\epsilon(0)\cap V_S$ for $|S|\le d-1$ where $V_S$ is the intersection of varieties of  $V_i$ for all $i\in S$, let $f_{k_i}$ be the corresponding irreducible polynomial. So we only need to check whether the numerator of $AA^+b-b$, denoted by $g(\delta)$, is zero on $V_S$. Since $f_{k_i}$ are irreducible, this is equivalent to whether there exists $i\in S$ such that $f_{k_i}$ divides $g(\delta)$.
\end{enumerate}

The overall algorithm can be expressed as the following. 
\begin{algorithm}
\caption{The local feasibility of $A(\delta)x=b(\delta)$}\label{alg:equality-case}
\begin{algorithmic}[1]
\State Calculate $\text{det}(A^TA+y I)=\sum_{i=0}^k y^iq_i(\delta)$
\State Factorize every $q_i$ and get a set of irreducible factors in $\mathbb{Q}[\delta_1,\delta_2,\cdots,\delta_d]$, $f_1,f_2,\cdots f_r$
\State Pick factors that are 0 when $\delta=0$, suppose they are $f_{k_1},\cdots,f_{k_{r'}}$. Denote their varieties by $V_1,\cdots,V_{r'}$
\State Calculate the numerator of $AA^+b-b$ on $B_\epsilon(0)\cap V_S$ for $|S|\le d-1$, on $B_\epsilon(0)\backslash (\cup_{|S|\le d,S\subset [r]}V_S)$, and on 0 by \eqref{eq:equality-pseudo-inverse-with-region}. Denote it by $g(\delta)$.
\State Check the following three conditions:
\begin{itemize}
    \item $g(0)=0$
    \item $g(\delta)=0$ on $B_\epsilon(0)\backslash (\cup_{|S|\le d,S\subset [r]}V_S)$
    \item $\forall S\subset [r']$, $\exists i\in S$ that $ f_{k_i}$ divides  $g(\delta)$ on $V_S$
\end{itemize}
\If{all conditions are met}
\State \textbf{return} \textbf{Feasible}
\Else
\State \textbf{return} \textbf{Not feasible}
\EndIf
\end{algorithmic}
\end{algorithm}

\newpage

\appendix

\section{Details of Lemma~\ref{lem:pca-transition-of-potential} and Lemma~\ref{lem:key-makur20}}\label{app:matrix}

In this section we describe the exact expression of matrices used in Lemma~\ref{lem:pca-transition-of-potential} and Lemma~\ref{lem:key-makur20} in both application to PCA and application to broadcasting on grid.

In Lemma~\ref{lem:pca-transition-of-potential}, the matrix $C_\ell$ is defined to be 
    \begin{align*}
        C_\ell(s,s') &= \Pr[(\eta_{t+1}(1),\cdots, \eta_{t+1}(\ell))=s'|(\eta_{t}(1-a),\cdots, \eta_{t}(\ell+b))=s]\\
        &= \prod_{n=a+1}^{\ell+a+1} \Pr[f(s(n-a),s(n-a+1),\cdots,s(n+b))=s'(n-a)],
    \end{align*}
In all our applications, $a=1$ and $b=0$, so $c_\ell$ is a $3^{\ell+1}\times 3^\ell$ matrix. $C_1$ is simply the transition matrix. For PCA with NAND function and vertex noise, we can get from the definition that $C_1$ is Table~\ref{tab:PCA-vertex}.
\begin{table}[]
\begin{center}
\begin{tabular}{ |m{1cm}||m{1cm}|m{1cm}|m{1cm}| } 
 \hline
  & 0 & 1 &?  \\ 
  \hline
 00& $p$ & $1-p $& 0\\
 01 & $p$ & $1-p$ & 0\\
 0? & $p$ & 1-$p$ & 0\\
 10 & $p$ & $1-p$ & 0\\
 11 &$1-p$ & $p$ & 0\\
 1? & $p$ & $p$ & $1-2p$\\
 ?0& $p$ & $1-p$ & $0$\\
 ?1 & $p$ & $p$ & $1-2p$\\
 ?? & $p $& $p$ & $1-2p$\\
 \hline
\end{tabular}
\end{center}
\caption{Transition matrix for PCA with NAND function and vertex noise}\label{tab:PCA-vertex}
\end{table}
For PCA with NAND function and edge noise, or broadcasting of in formation on grid with NAND function, $C_1$ is Table~\ref{tab:PCA-edge}.

For broadcasting of information on the grid with IMP function, $C_1$ is given by Table~\ref{tab:broadcasting-IMP}.

\begin{center}
\begin{table}[]
    \centering
\begin{tabular}{| m{1cm}||m{2cm}|m{2cm}|m{2cm}| } 
 \hline
  & 0 & 1 &?  \\ 
  \hline
 00 & $p^2$ & $1-p^2$ & 0\\
  01 & $p-p^2$ & $1-p+p^2$ & 0\\
 0? & $p^2$ & $1-p+p^2$ & $p-2p^2$\\
 10 & $p-p^2$ & $1-p+p^2$ & 0\\
 11 &$1-2p+p^2$ & $2p-p^2$ & 0\\
 1? & $p-p^2$ & $2p-p^2$ & $1-3p+2p^2$\\
 ?0& $p^2$ & $1-p+p^2$ & $p-2p^2$\\
 ?1 & $p-p^2$ & $2p-p^2$ & $1-3p+2p^2$\\
 ?? & $p^2$ & $2p-p^2$ & $1-2p$\\
 \hline
\end{tabular}
    \caption{Transition matrix for PCA with NAND function and vertex noise, and for broadcasting on grid with NAND function}
    \label{tab:PCA-edge}
\end{table}
\end{center}

\begin{table}
\begin{center}
\begin{tabular}{| m{1cm}||m{2cm}|m{2cm}|m{2cm}| } 
 \hline
  & 0 & 1 &?  \\ 
  \hline
 00 & $p-p^2$ & $1-p+p^2$ & 0\\
  01 & $p^2$ & $1-p^2$ & 0\\
 0? & $p^2$ & $1-p+p^2$ & $p-2p^2$\\
 10 & $1-2p+p^2$ & $2p-p^2$ & 0\\
 11 &$p-p^2$ & $1-p+p^2$ & 0\\
 1? & $p-p^2$ & $2p-p^2$ & $1-3p+2p^2$\\
 ?0& $p-p^2$ & $2p-p^2$ & $1-3p+2p^2$\\
 ?1 & $p^2$ & $1-p+p^2$ & $p-2p^2$\\
 ?? & $p^2$ & $2p-p^2$ & $1-2p$\\
 \hline
\end{tabular}
\end{center}
\caption{tab:Transition matrix for broadcasting on grid with IMP function}
\label{tab:broadcasting-IMP}
\end{table}

For arbitrary $\ell$, we can define 
\begin{equation*}
    C_\ell(s,s')= \prod_{i=1}^\ell C_1((s_i,s_{i+1}),(s'_i)).
\end{equation*}
Note that $C_\ell$ is a matrix polynomial of $p$ in all the settings we discuss.

For Lemma~\ref{lem:key-makur20}, the matrices are defined through a graph that connects patterns that appear consecutively in a string. Let $G$ be a directed weighted graph  with length $l$ strings as vertices and there is an edge between two vertices $s_0$ and $s_1$ if $(s_0(1),s_0(2),\cdots s_0(l-1))=(s_1(0),s_1(1),\cdots s_1(l-2))$, where $s(i)$ stands for the $i$th symbol of $s$. For a potential functions $w=\sum_{s\in S}c_s\{s\}$ with length $\ell$, each edge $(s,s')$ has weight $c_s$ if $s\in S$ and 0 otherwise. This way, any string $s$ would corresponds to a circle on the graph with $i$th node being $(s_{i\mod \ell },s_{i+1\mod \ell },\cdots ,s_{i+\ell\mod \ell })$. $w(s)$ is equal to weight of the cycle defined by the sum of weight of all nodes on the cycle. $w(s)\ge 0$ for any $s$ would be equivalent to having no negative cycles.

Let $B_{in}$ and $B_{out}$ be matrices of size $3^\ell\times 3^{\ell+1}$ with vertices of $G$ being rows and edges being columns. Entry of $B_{in}$ at row $v$ and column $e$ is 1 if the corresponding $e$ points to $v$ and 0 otherwise. Similarly, $B_{out}$ has entry 1 if $e$ points out from $v$. 

We can show that to show a directed graph having no negative cycle, it is sufficient to show that 
$$(B_{out}-B_{in})^Tz\le B_{out}^Tc$$
where $c$ is the vector of weights of all edges. So $A_\ell$ in Lemma~\ref{lem:key-makur20} is $(B_{out}-B_{in})^T$ and $B_\ell$ is $B_{out}^T$.

\section{Proof of of Lemma~\ref{lem:pca-aux-chain-edge}}\label{sec:pca-aux-chain-edge}

To simplify the discussion, we extend the definition of $NAND$ and BSC channel on alphabet $\{0,1,?\}$ as follows.
\begin{center}
\begin{tabular}{|c|c| } 
 \hline
 NAND & result  \\ 
  \hline
 0* & 1\\
 11 &0\\
 1? & ?\\
 ?1 & ?\\
 ?? & ?\\
 \hline
\end{tabular}
\end{center}

\begin{center}
\begin{tabular}{| m{0.7cm}||m{1.5cm}|m{1.5cm}|m{1.5cm}| } 
 \hline
  BSC& 0 & 1 &?  \\ 
  \hline
 0 & $1-p$ & $p$ & 0\\
  1 & $p$ & $1-p$ & 0\\
 ? & $p$ & $p$ & $1-2p$\\
 \hline
\end{tabular}
\end{center}
So the chain of $F_p'$ can still be written as 
\begin{equation*}
    \eta_{t+1}(n)=NAND(BSC(\eta_t(n)),BSC(\eta_t(n-1)))
\end{equation*}
Here the only randomness comes from the extended BSC, which we refer to as edge noise.

\begin{definition}
    We define a partial order between symbols $0< ?> 1$. The order between configurations are done by entry-wise comparison, and the order between distribution of configurations are defined by stochastic domination.
\end{definition}

\begin{lemma}\label{lem:auxilary-monotone}
    For two configurations $\eta$ and $\eta'$ such that $\eta\ge \eta'$, the order remains after evolution under the same realization of edge noises.

    For two distributions $\mu$ and $\nu$ on configurations, if $\mu\ge \nu$, then $F_p'\mu\ge F_p'\nu$.
\end{lemma}
\begin{proof}
The first statement is easy to check based on the definition of $NAND$ and $BSC$.

For the second statment,
since $\mu\ge \nu$, there exists a coupling $(\eta,\eta')$ between the distributions such that $\eta\ge \eta'$ always holds. We can couple them using same realization of edge noise, we have $F_p'\delta_{\eta}\ge F_p'\delta_{\eta'}$. So $F_p'\mu\ge F_p'\nu$.
\end{proof}

\begin{proof}[Proof of Lemma~\ref{lem:pca-aux-chain-edge}]
Let us first extend the time range of the PCA to $\mathbb
{Z}$. 
Each BSC has probability $p$ being 0, probability $p$ being 1 and probability $1-2p$ being the original symbol.
Let $\tau$ be the random outcome of extended BSCs at all time and location, and define a random configuration $\eta^*$ as follows. $\eta^*(i)$ is 1 (or 0) if $\exists t>0$ such that starting with arbitrary configuration on $\eta_{-t}$ and let the PCA evolve based on $\tau$, $\eta_0(i)$ would be 1 (or 0). Otherwise $\eta^*(i)$ is ?.

Note by Lemma~\ref{lem:auxilary-monotone}, the order between configurations remains under same edge noise. Let $\eta^?$ be the result of $\eta_0$ with $\eta_{-t}$ being all ?, and $\eta^a$ be the result of $\eta_0$ with $\eta_{-t}$ being an arbitrary configuration. We have $\eta^?\ge \eta^a$. So $\eta^?(i)$ is 1 (or 0) imply that $\eta^a(i)$ being 1 (or 0).

Therefore, $\eta^*$ can be equivalently defined by the following limit. Let $\eta^t$ be that $\eta^t(i)$ is 1 (or 0) if starting with a configuration with all ? on $\eta_{-t}$ and let the PCA evolve based on $\tau$, $\eta_0(i)$ would be 1 (or 0). Otherwise $\eta^t(i)$ is ?. And $\eta^*=\lim_{t\rightarrow \infty}\eta^t$. This means starting with all ?, the PCA converges to $\eta^*$. So $\eta^*$ is a stationary distribution and has 0 probability of containing ?. 

Suppose we start with arbitrary configuration, by Lemma~\ref{lem:auxilary-monotone}, $\lim_{t\rightarrow \infty}\eta_t\le \eta^*$. When $\eta^*$ have no ?, the only configuration $\eta$ satisfying $\eta\le \eta^*$ is itself. So we proved that arbitrary configuration would converge to $\eta^*$.
\end{proof}

\section{Potential Functions for Ergodicity of PCA}\label{sec:potential-pca}
We give an explicit solution of the potential function for Theorem~\ref{thm:pca} and Theorem~\ref{thm:pca-edge}. The following is a feasible potential function that satisfies \eqref{eq:PCA-final-PLP} and proves Theorem~\ref{thm:pca}.
\begin{align}\label{eq:pot_pca}
\begin{split}
    w&=(-4+4\delta)\{000\}+(-8+4\delta)\{001\}+(-1-2\delta)\{00?\}+(2+4\delta)\{010\}\\
    &+(-2+4\delta)\{011\}+(1+5\delta/2)\{01?\}+(-2+4\delta)\{0?0\}+(-3-13\delta)\{0?1\}\\
    &+(-2-6\delta)\{0??\}+(-6+4\delta)\{100\}+(-10+4\delta)\{101\}+3\{10?\}\\
    &+4\delta\{110\}+(-4+4\delta)\{111\}+(-1-2\delta)\{11?\}+(-3-21\delta)\{1?0\}\\
    &-12\{1?1\}+(-5+7\delta/2)\{1??\}+(8-8\delta)\{?00\}+13\{?01\}\\
    &+21\{?0?\}+(8+3\delta)\{?10\}+(4-4\delta)\{?11\}+(7-16\delta)\{?1?\}\\
    &+(6-26\delta)\{??0\}+(4-15\delta)\{??1\}+(6-43\delta)\{???\}\\
\end{split}
\end{align}
The following is a feasible potential function that proves Theorem~\ref{thm:pca-edge}.
\begin{align}\label{eq:pot_pca-edge}
\begin{split}
    w&=(-2+\delta)\{000\}+(-2+\delta)\{001\}+(-4-\delta)\{00?\}+(-2+\delta)\{010\}\\
    &+(-2+\delta)\{011\}+(-9-25\delta)\{01?\}-2\{0?0\}+(5-52\delta)\{0?1\}\\
    &+(3-28\delta)\{0??\}+(-2+\delta)\{100\}+(-2+\delta)\{101\}+(-4+\delta)\{10?\}\\
    &+(-2+\delta)\{110\}+(-2+\delta)\{111\}+(-10 -45055\delta/2496)\{11?\}+(6 -4275059\delta/87360)\{1?0\}\\
    &-(12 -88969\delta/2496)\{1?1\}+(10 -65239\delta/5824)\{1??\}\\
    &+(2 +374207\delta/33280)\{?00\}+(2+ 440767\delta/33280)\{?01\}\\
    &+284767\delta\{?0?\}/33280+(-2 -1877\delta/210)\{?10\}\\
    &+(-3 -28031\delta/8736)\{?11\}+(-12+ 242057\delta/8736)\{?1?\}\\
    &+(-4 -9596961\delta/232960)\{??0\}+(2 -786109\delta/29120)\{??1\}-15311\delta\{???\}/4368\\
\end{split}
\end{align}

\section*{Acknowledgements}
GB was supported in part by NSF CAREER award CCF-1940205. This material is based upon work supported by the National Science Foundation under Grants No CCF-1717842, CCF-2131115.

 \newpage
\addcontentsline{toc}{section}{References}

\bibliographystyle{alpha}
 \bibliography{ref}

\end{document}